\documentclass[11pt, one sided]{amsart}
\RequirePackage[colorlinks,citecolor=blue,urlcolor=blue]{hyperref}

\usepackage{latexsym,amssymb,amsmath,amsfonts,amsthm}
\usepackage{amsthm,amsmath}
\usepackage{float}
\usepackage{enumerate}
\usepackage[margin=3cm]{geometry}
\usepackage{bbm}
\usepackage{subfigure}
\usepackage{graphicx}
\usepackage[toc,page]{appendix}
\usepackage{multirow}
\usepackage{tikz}

\usepackage{comment}

\usepackage{amsfonts}
\usepackage[all]{xy}
\usepackage{caption}
 \usepackage{geometry}                
\usepackage{graphicx}
\usepackage{amssymb}
\usepackage{epstopdf}
\usepackage{color}
\usepackage{bbm, dsfont}
\usepackage{hyperref}
\usepackage[utf8]{inputenc}
\usepackage{amssymb,amsthm,amsmath,amssymb,wrapfig,dsfont}

\usepackage{tikz}
\usetikzlibrary{decorations.pathmorphing}
\tikzset{snake it/.style={decorate, decoration=snake}}
\usetikzlibrary{shapes.geometric,positioning,decorations.pathreplacing} 
\usepackage{pgfplots}

\newtheorem{theorem}{Theorem}

\newtheorem{conjecture}{Conjecture}

\def\beq{ \begin{equation} }
\def\eeq{ \end{equation} }

\def\square{\vcenter{\vbox{\hrule height .4pt
  \hbox{\vrule width .4pt height 5pt \kern 5pt
        \vrule width .4pt} \hrule height .4pt}}}

\newcommand{\bae}{\begin{equation}\begin{aligned}}
\newcommand{\eae}{\end{aligned}\end{equation}}

\DeclareFontFamily{OML}{rsfs}{\skewchar\font'177}
\DeclareFontShape{OML}{rsfs}{m}{n}{ <5> <6> rsfs5 <7> <8> <9>
rsfs7 <10> <10.95> <12> <14.4> <17.28> <20.74> <24.88> rsfs10 }{}
\DeclareMathAlphabet{\mathfs}{OML}{rsfs}{m}{n}

\newcommand{\note}[1]{{\color{red}{ \bf{ [Note: #1]}}}}

\usepackage[nomessages]{fp}

\begin{document}
\pgfmathsetseed{2}
\title{A Geometric Picture of perturbative QFT}

\author{Rory O'Dwyer}
\address[Rory O'Dwyer]{Department of Physics, Stanford University}

\maketitle
\begin{abstract}

In previous work, a lattice scalar propagator was rigorously defined in $d=1$ flat space and shown to equal the known Klein-Gordon propagator of QFT. This work generalizes this lattice propagator to manifolds whose universal cover is the hyperbolic half plane as well as a broad class of higher dimensional manifolds. We motivate a conjecture for the power spectrum of these curved space propagators. Afterwards, an analogous definition of the Dirac spinor propagator is defined. The formulation of these propagators are unified by the understanding of the object as a Fourier transform of the volume of path space of paths with the same length $I$ to mass $m$; the main theorem of this work will be to demonstrate that the point correlators of scalar perturbative QFT can be understood as a Fourier transform of the volume of path space of paths, which realize a Feynman diagram and have some total length I. After demonstrating this theorem, the author goes on to conjecture the point correlators of the Abelian Higgs Model in this geometric formulation.


\end{abstract}



\section{Introduction}
\label{section:intro}

Previous work defined the relativistic lattice scalar propagator in $d=1$ and showed that it converged the the result expected from QFT \cite{odwyer2023relativistic}. It is the purpose of this work to first extend the relativistic lattice scalar propagator to higher dimensions and to manifold which have $\mathbb{H}$ as their universal cover. We will also show how our definition of the propagator may be defined over notions of connected sums and surgeries \cite{defreitas2022geometrization}, extending the definition further.


The lattice scalar propagator is defined on graphs that look locally like $\mathbb{Z}^{d}\times \mathbb{Z}$. To define it, we first would count discrete paths between points in this graph, and partition them by a special minkowski metric like function. The partition function would be a path indexed sum over elements of the unit circle whose phase was determined by the length of said path. Then, a special continuum limit was taken, obtained first in \cite{continuousbin} and \cite{continuouslatticepath}, and this was shown to converge to the known form of the KG propagator in $d=1$.

If we have some discretization of a 2d manifold, this definition may normally be extended to said context. In the instance that \cite{odwyer2023relativistic} defined the discrete propagator on $\mathbb{T}^{2}$, it used the fundamental domain of $\mathbb{T}^{2}$, a square, to tile $\mathbb{Z}^{2}$. It could then sum over discrete paths from an initial point to all lifts of the final point in $\mathbb{Z}^{2}$ to define the discrete propagator given a definition of it in $\mathbb{Z}^{d}\times \mathbb{Z}$. We not only aim to perform the same calculation here but to also do so for the full variety of 2d manifolds sans the two sphere and to obtain the continuum limit of these objects.


In this paper, we will make small adaptations to the proof of \cite{odwyer2023relativistic} so that it defines a discrete propagator over $\mathbb{Z}^{d}\times \mathbb{Z}$ and a conintuum one over $\mathbb{R}^{d}\times\mathbb{R}$. We will show two alternative ways to then define a discretization of general 2d manifolds. The first consists of using Teichmuller curves; these curves are topologically isomorphic to variuous $g\ge 2$ curves. They have a natural flat metric on them and are important in the study of billiards \cite{tiechmuller}; their structure is not necessarily conformally or diffeomorphically isomorphic to the 2d manifolds they describe as they have singular discontinuities in curvature. How we would extend a propagator over a surgery cut or a connected sum will become apparent in this section (Section~\ref{sec:tiech}).

Our second approach is more nuanced and consists of discretizing the hyperboloid model of hyperbolic geometry \cite{hyperboloidmodel}. By the uniformization theorem of Riemann Surfaces \cite{bujalance_costa_martinez_2001}, all compact Riemann surfaces with genus $g\ge 2$ have the half plane as their universal cover. This means that their fundamental domain will tile hyperbolic space. This will be covered in Section~\ref{sec:hyperboloid}.

After this, the author will show that the Fermionic propagator has a very similar geometric interpretation to the bosonic one (Section~\ref{sec:fermi}). Namely, he will show that it is a fourier transform of paths avoiding a direction dictated by the initial and final spinor states. An almost identical proof to the geometric propagator will proceed thereafter, and the author is confident that in this manner one may define the lattice fermionic propagator in a similarly wide context.

With these tools in hand, the jump to a geometric model of QFT can be made (Section~\ref{section:Interactions}). The author will show that the Feynman rules of perturbative QFT have an interpretation in this geometric picture. Specifically, the contribution from some Feynman diagram is equal to the Fourier transform of the volume of paths that realize that Feynman diagram through space; i.e. Feynman diagrams aren't just means of counting point correlator contributions but represent genuine particle trajectories akin to the original insight of Feynman.

\section{Notations and Statement of Theorems}
\label{section:notation}

We will employ the definition of a discrete space $X$ and its continuum partner $X^{cont}$ as defined in \cite{odwyer2023relativistic}. Namely, let $d\in\{1,...,\}$ and $\{L_{i}\}_{i=1}^{d}\subset \mathbb{N}\cup\{\infty\}$. I define $X=(\times_{i}\{-L_{i},-L_{i}+1,...,L_{i}\})\times \mathbb{Z}$. The first $d$ copies of $\mathbb{Z}$ are defined as the \textbf{spatial} coordinates and the last with \textbf{time}. I note that $X\subset \mathbb{Z}^{d}$, so it has a boundary $\partial X$ in $\mathbb{Z}^{d}$. I will constrain ourselves to $X$ such that $2|\left|\partial X\right|)$. I can therefore give $\partial X$ an arbitrary equivalence relation $\sim$ which partitions $\partial X$ into pairs. I shall leave $\sim$ completely general. In \cite{odwyer2023relativistic}, I only defined the polygonal metric for $d=1$; I define it in the upcoming section on all above X for arbitrary dimension d.



To define polygonal metrics on X, its appropriate to first define them on the square lattice. Let $\vec{x}_{1},\vec{x}_{2}\in \mathbb{Z}^{d}\times \mathbb{Z}$. I define on the following functions $d_{l_{2}},d_{l_{2}^{*}}: \mathbb{Z}^{d}\times \mathbb{Z}\rightarrow \mathbb{C}$

\begin{multline}
    \label{eq:minko}
    d_{l_{2}}(\vec{x}_{1},\vec{x}_{2})=\sqrt{proj_{t}(\vec{x}_{1}-\vec{x}_{2})^{2}+\sum_{i=1}^{d}proj(\vec{x}_{1}-\vec{x}_{2})^{2}}
    \\
    d_{l_{2}^{*}}(\vec{x}_{1},\vec{x}_{2})=\sqrt{proj_{t}(\vec{x}_{1}-\vec{x}_{2})^{2}-\sum_{i=1}^{d}proj(\vec{x}_{1}-\vec{x}_{2})^{2}}
\end{multline}

For all our purposes, I will only be using $d_{l_{2}^{*}}$ on $d_{l_{2}^{*}}^{-1}(\{a|a\in \mathbb{R},a\ge 0\})$. The first equation in Equation~\ref{eq:minko} is clearly the $l_{2}$ metric on $X$. The subscript $l_{2}^{*}$ in the second is so recognize its similarity with the $l_{2}$ metric; I will refer to the second part of Equation~\ref{eq:minko} as the \textbf{minkowski metric}. Let $n\in \{1,2,...\}$. Consider the set of scaled primitive \textbf{pythagorean tupples} with hypotenuse below n in Equation~\ref{eq:pythagTupple}.

\begin{multline} 
    \label{eq:pythagTupple}
    \mathcal{A}_{n}=\{(\frac{x_{1}}{I},...\frac{x_{d}}{I},\frac{t}{I})\in\mathbb{Q}^{d+1}|\sum_{i=1}^{d}x_{i}^{2}+I^{2}=t^{2},(x_{1},...,x_{d},I,t)\in\mathbb{Z}^{d+2},0\le t\le n,I\neq0\}
    \\\cup\{\frac{1}{gcd(x_{1},...,x_{d},t)}(x_{1},..,x_{d},t)\in \mathbb{Q}^{d+1}|\sum_{i=1}^{d}x_{i}^{2}+t^{2}=0,(x_{1},...,x_{d},t)\in\mathbb{Z}^{d+1},0\le t\le n\}
\end{multline}

The word `primitive' places an extra constraint on $\mathcal{A}_{n}$ that if we consider the equivalence class placed upon these vectors by parallelism, then we only keep the representative with lowest hypotenuse. If we want a generalization of the polygonal metric developed in \cite{odwyer2023relativistic}, we note that the equation $\{\vec{x}\in \mathbb{R}^{d+1}|d_{l_{2}^{*}}(0,\vec{x})=1\}$ is a hyper-surface of dimension d, so we will want a polyhedron of dimension d approximating this quadratic surface to be the unit sphere of our polygonal metric. Let $\vec{v}\in \mathcal{A}_{d}$. Then, we define a \textbf{neighborhood} $\mathcal{N}_{\vec{v}}$ of $\vec{v}$ to be one choice of the $d$ closest linearly independent vectors in $\mathcal{A}_{d}$ to $\vec{v}$. We specify this choice to allow for the case that there are multiple d tupples of vectors satisfying this requirement. If we have some $\mathcal{N}_{\vec{v}}$, let $\vec{v}^{\perp}$ be a unit vector (as defined by $d_{l_{2}}$) which represents the (atmost) one dimensional space $\{\vec{w}-\vec{v}|\vec{w}\in\mathcal{N}_{v}\}^{\perp}$. Then, we define the half-space function $\mathcal{H}_{\mathcal{N}_{\vec{v}}}:\mathbb{R}^{d}\rightarrow \mathbb{R}$ as in Equation~\ref{eq:halfspace}.

\begin{equation}
    \label{eq:halfspace}
    \mathcal{H}_{\mathcal{N}_{\vec{v}}}(\vec{x})=\frac{\sum_{i=1}^{d}(proj_{x_{i}}(\vec{v}^{\perp})proj_{x_{i}}(\vec{x}))+proj_{t}(\vec{v}^{\perp})proj_{t}(\vec{x})}{\sum_{i=1}^{d}(proj_{x_{i}}(\vec{v}^{\perp})proj_{x_{i}}(\vec{v}))+proj_{t}(\vec{v}^{\perp})proj_{t}(\vec{v})}
\end{equation}

This function is so defined because it will trace out a $d$ dimensional half-space (the intersection of which would give us a d dimensional polyhedron); it agrees with $d_{l_{2}^{*}}$ on setting $\vec{v}$ to length 1, and its roughly tangent to the point $\vec{v}$. It yields the candidate expression for the \textbf{d dimensional polygonal minkowski metric} in Equation~\ref{eq:candidate}.

\begin{equation}
    \label{eq:candidate}
    d_{n}(\vec{x}_{1},\vec{x}_{2})=min(\mathcal{H}_{\mathcal{N}_{\vec{v}}}(\vec{x}_{2}-\vec{x}_{1})|\vec{v}\in\mathcal{A}_{n})
\end{equation}

One can see that these functions approximate the minkowski metric in $d$ dimensional space in Theorem~\ref{thm:axes} and via Figure~\ref{fig:polygonalMet}. 

\begin{figure}[H]
    \centering
    \includegraphics[width = 8 cm]{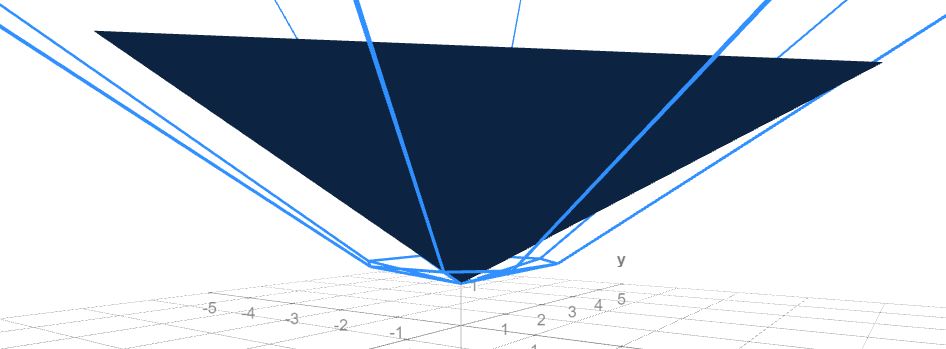}
    \caption{The set $\{d_{n}(0,\vec{x})=1\}$ for $d=2$ and $n=1$ and $n=3$}
    \label{fig:polygonalMet}
\end{figure}

As you let $n\rightarrow\infty$, you see the unit circles for $d_{n}$ in Figure~\ref{fig:polygonalMet} approximate a hyperboloid of two sheets, which is the unit circle for $d_{l_{2}^{*}}$ in this dimension. The definition obtained for $d_{n}$ in this setting, as opposed to \cite{odwyer2023relativistic}, generalizes to arbitrary finite dimension d. Now, for $d=2$ (and in higher dimensional generalizations of discrete orbifold spaces), the extension of our metric to $X$ is rather simple. Ignoring the temporal component of $X$, we note that $P$ lifts via the boundary tiling to the covering lattice $\mathbb{Z}^{2}$. If we want to find the distance between two points in X, we first find the lifts of each point which have the smallest graph distance in $\mathbb{Z}^{2}$. Then, we define the distance between our original points in $X$ as the distance $d_{n}$ on these closest lifted points in $\mathbb{Z}^{2}$.

I shall return back to Equation~\ref{eq:minko} and Equation~\ref{eq:candidate} for a moment. Let's define $\mathcal{A}^{all}=\{\vec{x}\in \mathbb{Z}^{d}\times\mathbb{Z}|d_{l_{2}^{*}}(0,\vec{x})=d_{n}(0,\vec{x})\ge 0,proj_{t}(\vec{x})>0\}$. A subset $\mathcal{A}^{gen}\subset \mathcal{A}^{all}$ is called generating iff for all $a\in \mathcal{A}^{all}$ there is $\{a_{i}\}_{i=1}^{N}\subset \mathcal{A}^{gen}$ such that $\sum_{i=1}^{N}=a$. Then, an \textbf{axes of symmetry} of $d_{n}$, denoted $\mathcal{A}_{n}$, is some $\mathcal{A}_{gen}$ such that if $\vec{x}_{1},\vec{x}_{2}\in \mathcal{A}_{n}$ and $\vec{x}_{1}\neq c\vec{x}_{2}$ then $c=1$. So, every direction has a unique representative. The double usage of $\mathcal{A}_{n}$ for scaled Pythagorean tuples and as an axes of symmetry of $d_{n}$ will be shown to not be an abuse of notation in Theorem~\ref{thm:axes} (except the axes of symmetry of $d_{n}$ will include also light paths). Our language in this paper will suggest that this set is unique to each $d_{n};$ we will not prove this for we do not require it for the next concepts.



I define as a \textbf{path} as set $\gamma=\{\vec{x}_{i}\}_{i=1}^{n}$ such that $\{x_{i}x_{i+1}\}_{i=1}^{n-1}\subset\mathcal{E}$. Then, our \textbf{achronal, local paths} between $\vec{x}$ and $\vec{y}\in X$ as the following:

\begin{equation}
    \label{eq:pathset}
    \Gamma^{\vec{x},\vec{y}}=\{\gamma|\gamma|_{0}=\vec{x},\gamma|_{n}=\vec{y},d_{l_{2}^{*}}(\vec{x}_{i},\vec{x}_{i+1})\ge 0,d_{l_{2}^{*}}(\vec{x},\vec{y})\in \mathbb{Z},proj_{t}(\vec{x}_{i+1},\vec{x}_{i})\ge 0\}
\end{equation}

I define our \textbf{achronal, local $n$ polygonal paths} between $\vec{x}$ and $\vec{y}\in X$ as the following:

\small
\begin{equation}
    \label{eq:pathset2}
    \Gamma_{n}^{\vec{x},\vec{y}}=\{\gamma|\gamma|_{0}=\vec{x},\gamma|_{n}=\vec{y},d_{l_{2}^{*}}(\vec{x}_{i},\vec{x}_{i+1})=d_{n}(\vec{x}_{i},\vec{x}_{i+1})\ge 0,d_{n}(\vec{x}_{i},\vec{x}_{i+1})\in\mathbb{Z},proj_{t}(\vec{x}_{i+1},\vec{x}_{i})\ge 0\}
\end{equation}
\normalsize

On $\Gamma_{n}$ and $\Gamma$ we place an equivalence relation on paths; two paths are equivalent if their indexing set trace out the same piecewise linear paths in $X$. The use of an axes of symmetry is made clear in this context; each equivalence class in $\Gamma_{n}^{\vec{x},\vec{y}}$ has a unique representative with difference sequence drawn from $\mathcal{A}_{n}$ (Theorem~\ref{thm:generate}). This paper will sometimes refer to sets of paths $\Gamma$ as \textbf{path spaces}.

From Equation~\ref{eq:pathset} and Equation~\ref{eq:pathset2} I have $\Gamma^{\vec{x},\vec{y}}_{n}\subset \Gamma^{\vec{x},\vec{y}}$. Let $\gamma\in \Gamma^{\vec{x},\vec{y}}$. Then, I define the polygonal and minkowski \textbf{proper time} in Equation~\ref{eq:inter}.

\begin{equation}
    \label{eq:inter}
    \rho_{n}(\gamma)=\sum_{i=0}^{n-1}d_{n}(x_{i+1},x_{i}),\rho_{l_{2}^{*}}(\gamma)=\sum_{i=0}^{n-1}d_{l_{2}^{*}}(x_{i+1},x_{i})
\end{equation}

From this definition, I immediately obtain Theorem~\ref{thm:axes}. These geometrical results are developed because I will be viewing the path integral as a geometrical object. For this purpose, I define the following functions $K_{n}:X\times X\rightarrow \mathbb{C}$ and $K_{l_{2}^{*}}:X\times X\rightarrow \mathbb{C}$. Let $m\in \{a|a\in \mathbb{R},a>0\}$.

\begin{equation}
    \label{eq:propagators}
    K_{n}(\vec{x},\vec{y})=\sum_{\gamma\in\Gamma_{n}^{\vec{x},\vec{y}}}e^{im\rho_{n}(\gamma)},
    K_{l_{2}^{*}}(\vec{x},\vec{y})=\sum_{\gamma\in\Gamma^{\vec{x},\vec{y}}}e^{im\rho_{l_{2}^{*}}(\gamma)}
\end{equation}

These functions are our \textbf{discrete propagators}, as initially defined in \cite{odwyer2023relativistic}. In this original paper, we introduced the \textbf{discrete Feynman Propagator} $K^{Feyn}_{n}$ as Equation~\ref{eq:propagators} except that we would allow $\rho_{n}(\gamma)$ to include difference sequence steps with both $\pm d_{n}(x_{i+1},x_{i})$ as their contribution to the Minkowski length. We will again make this definition.

These are all our basic definitions in the discrete context. I will devote some of our paper to the continuous setting. For this, I must define for our original X a natural domain $X^{cont}$ which is the result of taking infinitely fine lattices. We note that the continuum spaces in general are the \textbf{tiechmuller curves} mentioned in Section~\ref{section:intro}; curves that are homeomorphic but not conformally isomorphic to some 2d manifold \cite{tiechmuller}. $X^{cont}=(\times_{i}-[L_{i},L_{i}])\times \mathbb{R}$ with an associated equivalence relation $\sim$ on pairs of points in $\partial X^{cont}$ as a subspace of $\mathbb{R}^{d}\times \mathbb{R}$. I need a generalization of the closest integer function $[*]:X^{cont}\rightarrow X$. For $L\in \mathbb{N}\cup\{\infty\}$ I let $[*]^{1d}:[-L,L]\rightarrow \{-L,-L+1,...,L\}$ be the closest point in the range to the domain under the $l_{2}$ metric.

Then our generalization is defined as follows:

\begin{multline}
    \label{eq:closest}
    [\vec{x}]\textrm{ is the element of }X\textrm{ such that }
    \\
    proj_{x_{i}}([\vec{x}])=[proj_{x_{i}}(\vec{x})]^{1d}
    \\
    proj_{x_{i}}([\vec{x}])=[proj_{t}(\vec{x})]^{1d}
\end{multline}

Finally, for those $X$ having a finite element $L_{i}$ in their definition, I must also change their dimensions to approach $X^{cont}$ as a limiting space for infinitely fine lattices. Therefore, when I write $\Gamma_{n}|_{m},\Gamma|_{m},K_{n}|_{m}$, or $K_{l_{2}^{*}}|_{m}$, I am referring to those objects for $X|_{m}=(\times_{i}\{-mL_{i},-mL_{i}+1,...,mL_{i}\})\times \mathbb{Z}$ when $X=(\times_{i}\{-L_{i},-L_{i}+1,...,L_{i}\})\times \mathbb{Z}$. 

Counting lattice paths are naturally connected to the multinomial coefficient; therefore, continuum propagators seem to require the development of a continuous multinomial coefficient. As defined by Cano and Diaz in \cite{continuousbin}, and later generalized by Wakhare,Vignat,Le, and Robins in \cite{continuouslatticepath}, we consider the continuous multinomial coefficient in Equation~\ref{eq:contmult} for l variables. Let $\{x_{i}\}_{i=1}^{l}\subset \mathbb{R}_{+}$, $\mu$ denote the Lebesgue measure on $\mathbb{R}^{n}$ \cite{folland}, and for $n,N\in \mathbb{N}$ denote  by $D(n,N)$ the number of Smirnov words of length n and N \cite{continuouslatticepath}. Furthermore, let $c\in D(n,N)$, let $\{c_{k}\}_{k=1}^{n}$ denote singular letters of our Smirnov word taking elements among some collection of d dimensional vectors, and let $\{e_{k}\}_{k=1}^{l}$ denote these vectors. For $q\in \mathbb{R}_{+}^{d}$ we define the \textbf{path polytope} $P(q,c)$ in Equation~\ref{eq:pathpoly}.

\begin{equation}
    \label{eq:pathpoly}
    P(q,c)=\{(\lambda_{1},...,\lambda_{n})\in \mathbb{R}_{+}^{n}|\sum_{k=0}^{n}\lambda_{k}e_{c_{k}}=q\}
\end{equation}

Then, we have our desired definition of the continuous multinomial in Equation~\ref{eq:contmult}:

\begin{equation}
    \begin{Bmatrix}\sum x_{i}\\x_{1},x_{2},...,x_{l}\end{Bmatrix}=\sum_{n=0}^{\infty}\sum_{c\in D(n,l)}\mu (P(q,c))
    \textrm{ where }q=\sum x_{i}e_{i}
    \label{eq:contmult}
\end{equation}

Problematically, the continuous multinomial was never shown by either \cite{continuousbin} or \cite{continuouslatticepath} to converge except in special cases. The necessary convergence results are developed in Section~\ref{section:def}; we will also show more rigorously in what sense they are continuous analogues of the discrete multinomial coefficients. Namely, for $m\in \mathbb{N}$, we introduce the operator $\mathcal{T}^{m}_{cont}$ in Equation~\ref{eq:operator}. This operator acts on sums indexed over paths in $\Gamma^{\vec{y},\vec{x}}$. Let $f:\Gamma^{\vec{y},\vec{x}}\rightarrow \mathbb{C}$ be some path indexed complex valued function, and let $r,\theta:\Gamma^{\vec{y},\vec{x}}\rightarrow\mathbb{R}$ be its radial and angular components. Let $|\gamma|$ denote the number of distinct linear segments in $\gamma$. Then, we can rearrange any general sum over $\gamma\in \Gamma^{\vec{y},\vec{x}}$ as done in Equation~\ref{eq:rearrange}.

\begin{equation}
    \sum_{\gamma\in\Gamma^{\vec{x},\vec{y}}}f(\gamma)=\sum_{n=1}^{\infty}\sum_{\gamma\in\Gamma^{\vec{x},\vec{y}},|\gamma|=n}f(\gamma)
    \label{eq:rearrange}
\end{equation}

Then, we may write $\mathcal{T}_{cont}^{m}$ easily in this context (in Equation~\ref{eq:operator})

\begin{equation}
    \mathcal{T}^{m}_{cont}(\sum_{\gamma\in\Gamma^{\vec{x},\vec{y}}}f(\gamma))=\sum_{n=d}^{\infty}(\sum_{\gamma\in\Gamma^{\vec{x},\vec{y}},|\gamma|=n}r(\gamma)m^{d-n}e^{i\frac{\theta(\gamma)}{m}})
    \label{eq:operator}
\end{equation}

As demonstrated in Theorem~\ref{thm:disctocont}, this operator is required to connect a natural estimate of counting paths to a notion of volume, as one would expect in the continuous setting. With all the relevant concepts defined, the following \textbf{continuum propagators} are defined in Equation (in the event they exist).

\begin{multline}
    \label{eq:cont}
    K_{n}^{cont}(\vec{x},\vec{y})=lim_{m\rightarrow\infty}\frac{\mathcal{T}_{cont}^{m}K_{n}([m*\vec{x}],[m*\vec{y}])|_{m}}{\mathcal{T}_{cont}^{m}(max_{\vec{x'},\vec{y'}\in X|_{m},proj_{t}(\vec{y}'-\vec{x}')=[nt]}(\left|\Gamma_{n}^{\vec{x'},\vec{y'}}|_{m}\right|))}
    \\
    K_{l_{2}^{*}}^{cont}(\vec{x},\vec{y})=lim_{p\rightarrow\infty}K_{p}^{cont}(\vec{X},\vec{y})
\end{multline}

In Equation~\ref{eq:cont}, we will be satisfied if the sequence of $K_{n}^{cont}$ has a convergent subsequence towards some function in a pointwise or sup norm sense, and therefore that this limit is unique w.r.t limits. We define $K_{n}^{Feyn,cont}$ to be Equation~\ref{eq:cont} but with each instance of $K_{n}$ replaced with $K_{n}^{feyn}$.

These are all the definitions we require for Theorem~\ref{thm:k3} and Theorem~\ref{thm:tiech}. For our path integral on hyperbolic space, we define an indexed set of hyperboloids in Equation~\ref{eq:hyperboloid}.

\begin{equation}
    \label{eq:hyperboloid}
    \mathcal{H}_{t_{0}}=\{(t,x,y)|t-t_{0}+1=\sqrt{1+x^{2}+y^{2}}\}
\end{equation}

Each $\mathcal{H}_{t_{0}}$ inherits the minkowski $d_{l_{2}^{*}}$ and polygonal $d_{n}$ metrics from $\mathbb{R}^{2}\times\mathbb{R}$.  With this metric defined, we can consider the group of isometries on $\mathcal{H}_{t_{0}}$. That group is known to be $PSL(2,R)$\cite{hyperboloidmodel}. If $G\subset PSL(2,R)$ is a subgroup, and if $\forall \vec{x}\in\mathcal{H}_{t_{0}}$ the orbit of x $Orb_{G}(x)=\{y\in\mathcal{H}_{t_{0}}|y=gx,g\in G\}$ has no cluster points under our metric this group is \textbf{Fuschian} \cite{bujalance_costa_martinez_2001}. It is a fundamental theorem of the geometrization of 2d manifolds that any 2d manifold whose uniform cover is not $\mathbb{S}^{2}$ or $\mathbb{R}^{2}$ has an isometric embedding into $\mathcal{H}_{t_{0}}$ where its homotopy group $\pi_{1}(X)$ is realized as a Fuschian group\cite{bujalance_costa_martinez_2001}. This allows us to define another $X$ and $X^{cont}$ which will extend the lattice scalar path integral to all 2d manifolds.

For what comes next, we want to redefine how $[*]$ works as the closest point in $\mathbb{Z}^{d}\times\mathbb{Z}$ to some point in $\mathbb{R}^{d}\times\mathbb{R}$.

We let $X^{cont}=\{\mathcal{H}_{t}|t\in\mathbb{R}\}=\mathbb{R}^{2}\times \mathbb{R}$ with an associated Fuschian group $G$ and equivalence equation $\sim$ such that $(t_{1},x_{1},y_{2})\sim(t_{1},x_{2},y_{2})$ if $(t_{2},x_{1},y_{2})\in Orb_{G}(t_{2},x_{2},y_{2})$. X, the discrete space, will be a $\mathbb{Z}^{2}\times \mathbb{Z}$ with a similar quotient structure. We would like to give $X$ an equivalence relation $\sim$ such that the closest point function $[*]$ would always find  for any point in its preimage the entire equivalence class for said point in the same preimage. Say $x,y\in\mathbb{Z}^{2}\times \mathbb{Z}$ both are the image of $[*]$ of two points in the same equivalence class $\sim$ for $X^{cont}$, and let $z\in\mathbb{Z}^{2}\times\mathbb{Z}$ have this same property w.r.t. y. The Fuschian group of $X^{cont}$ is a subgroup of the group of isometries of hyperbolic space. So, if there is a point in a geodesic ball about $y$ and one about $x$ that are mapped to each-other by $g\in G$, and another point about $z$ to which y maps to by $g_{2}$, then we can consider the image of the ball about x under $g_{2}g_{1}$. Within this image we will find a point in the ball around z that is the image of $g_{1}g_{2}$ of a point around a ball in x. If we consider $\sim$ on $\mathbb{Z}^{2}\times \mathbb{Z}$ defined st $x\sim y$ iff there is $x^{cont}\sim y^{cont}\in\mathbb{R}^{2}\times\mathbb{R}$ such that $[x^{cont}]=x$ and $[y^{cont}]=y$, then this relation is transitive. Furthermore, it is reflexive and symmetric, so this would define a valid equivalence relation on $\mathbb{Z}^{2}\times \mathbb{Z}$. 

That this equivalence relation has the desired property that it contains the whole orbit of $\vec{x}\in\mathbb{R}^{2}\times\mathbb{R}$ follows immediately from our argument above, and from the fact that the group action are isometries of this closest point metric.





\subsection{Fermionic Case}

Before moving to define the fermionic path integral in the lattice path integral picture, one must recognize some resemblance to Feynman's checkerboard. The Feynman checkerboard was an initial attempt to obtain the Fermion propagator as a genuine sum over paths done by Feynman in \cite{feynman1} and later generalized to other contexts in \cite{feynman2}. 

One of the nice properties the method presents here and \cite{odwyer2023relativistic} shares with Feynman's original calculation is a continuum procedure very similar to that instituted by $\mathcal{T}_{cont}$. Namely, Feynman also divided the set of lattice light paths into subsets indexed by their number of linear segments and weighed each linear segment just as is done by $\mathcal{T}_{cont}$. The procedure we adopted in \cite{odwyer2023relativistic}, while fundamentally different from Feynman's, is incidentally and encouragingly a `natural' way of approaching the Feynman path integral. This calculation had some peculiarities; namely, it assumed the Fermion traveled exclusively along light paths. We lose these peculiarities; our method uses light paths for a null set of trajectories and could extend to a definition of all perturbative QFT.


Let us move to define the fermionic propagator in the same geometric picture as we have for scalar particles. We define $\gamma^{\mu}$ as the $d+1$ many matrices constituting the algebra basis of the spin representation of the Clifford Algebra \cite{pesky} $Cl_{1,d}(\mathbb{R})$. The fermionic propagator always includes some two vectors $v,w\in\mathbb{C}^{2}$ describing the initial state of the fermion waveform; we define from these the special \textbf{initial frame} vector $\vec{V}=(v^{\dagger}\gamma^{\mu}w)\hat{e}_{\mu}-v^{\dagger}w\hat{I}$.

Now, consider the following argument. In previous work \cite{odwyer2023relativistic}, it was shown that the Klein-Gordon Propagator was the Fourier transform $I\rightarrow m$ of the number of directed paths with the same action $I$. This number of direct paths is the multinomial coefficient you will see in Theorem~\ref{thm:k3}. Let us assert that the fermionic path integral corresponds heuristically to the number of fermionic paths with some other (as of yet undetermined) action I. We now may introduce an action for our fermions in the form of Equation~\ref{eq:matlan} for some path $\gamma$ in space time.
\begin{equation}
    \label{eq:matlan}
\mathcal{L}_{\vec{v},\vec{w}}=\int_{\gamma}(md\tau+d\eta+id\theta)
\\
=\sum_{\vec{a}\in\gamma}(md_{l_{2}}^{*}(\vec{a})+\eta_{\vec{a}}+i\nu\theta_{\vec{a}})
\end{equation}

Here $\eta_{\vec{a}}$ denotes the rapidity of $\vec{a}$ w.r.t to the last frame. The angle $\theta_{\vec{a}}$ denotes the Wigner rotation of $\vec{a}$ w.r.t to the last frame. We expand upon their definitions in Equation~\ref{eq:rapidity}. We include $\nu\in \{0,1\}$ as an extra parameter as we will need contributions from paths both including and excluding the wigner rotation.

\begin{multline}
    \label{eq:rapidity}
     i>0\implies\eta_{\vec{a}_{i}}=arccosh(\frac{a_{i-1}*a_{i}}{d_{n}(0,a_{i-1})d_{n}(0,a_{i})})
     \\i=0\implies\eta_{\vec{a}_{i}}=arccosh(\frac{\vec{V}*a_{0}}{d_{n}(0,\vec{V})d_{n}(0,\vec{a}_{0})})
\\\theta_{\vec{a}}\textrm{ is the angle such that }
    \\cosh(\eta_{\vec{a}_{i}+\vec{a}_{i-1}})=cosh(\eta_{\vec{a}_{i}})cosh(\eta_{\vec{a}_{i}})+sinh(\eta_{\vec{a}_{i}})sinh(\eta_{\vec{a}_{i}})cos(a)
\end{multline}

Note that $a$'s definition (the Wigner rotation) is ambiguous as you can take it and $a+\pi$ as a solution to the definitional function. We will consider adding separate contributions for both of these terms (i.e. treat them as two separate paths just as we will do for the disparate values of $\nu$). Note that this, for each step in a single path, would give us 8 different terms (and actually 8 different paths). Then, our assertion is the following. We denote as \textbf{lattice fermionic propagator} $K_{n,\frac{1}{2}}(\vec{x},\vec{y})^{\vec{v},\vec{W}}$ the same path indexed sum over $\Gamma_{n}^{\vec{x},\vec{y}}$ as Equation~\ref{eq:propagators}. We denote the limit under $\mathcal{T}_{cont}$ of this discrete object as $K_{n,\frac{1}{2}}(\vec{x},\vec{y})^{cont,\vec{v},\vec{W}}$. We denote the limit of these continuous objects (regulated as in Equation~\ref{eq:cont}) to be $K_{\frac{1}{2}}(\vec{x},\vec{y})^{Feyn,cont,\vec{v},\vec{w}}$. This is the object that we assert will converge to the Dirac fermion propagator.

Note that the a superscript \textit{feyn} was added to this propagator, meaning path differences include spontaneous switches $m\rightarrow -m$ throughout the path. Just as with the $d>1$ scalar propagator, much of the arguments required to show Equation~\ref{eq:fermi4} was developed in previous work, and we need only mention where new arguments are required.



\subsection{Discussion on Interactions}

In this section, we introduce the manner in which perturbative QFT emerges from this picture and how it may lend itself to more analytical expressions for interactions. Consider the following interaction Lagrangian (Equation~\ref{eq:fourpole}).

\begin{equation}
    \label{eq:fourpole}
    \mathcal{L}=\int dx^{d}(\frac{1}{2}\partial_{\mu}\phi\partial^{\mu}\phi-m^{2}\phi^{2}+\lambda\phi^{3})
\end{equation}

Methods by Feynman \cite{pesky},\cite{Srednicki_2007} and others showed that this nonlinear interaction could be approximated by summing over convolutions of free propagators; these convolutions being indexed by Feynman diagrams (as in Figure~\ref{fig:contrib}).

\begin{figure}[h]
    \centering
    \includegraphics[width = 7 cm]{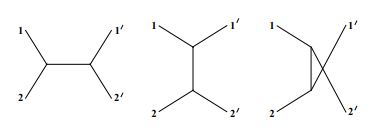}
    \caption{The 3 tree Level Contribution to the Four Field Correlator in $\lambda^{3}$ Theory \cite{Srednicki_2007}}
    \label{fig:contrib}
\end{figure}

The contribution you obtain from Figure~\ref{fig:contrib} is $\lambda^{2}$ times the convolution of three propagators. Namely, the propagator in with momentum $\vec{k}$, and integrating in the center over momenta $\vec{q}$ where the two paths in the loop take on $\vec{q}$ and $\vec{k}-\vec{q}$. This convolution is the Fourier transform of a product of free propagators; for instance, the leftmost diagram in Figure~\ref{fig:contrib} would give Equation~\ref{eq:contribution} as its contribution to the interaction propagator.

\tiny
\begin{equation}
    \label{eq:contribution}
    (i\lambda)^{2}(\frac{1}{i})^{5}\int d^{d}y d^{d}z K_{l_{2}^{*}}^{Feyn,cont}(y,z)(K_{l_{2}^{x}}^{Feyn,cont}(\vec{x}_{1},y)K_{l_{2}^{x}}^{Feyn,cont}(\vec{x}_{2},y)K_{l_{2}^{x}}^{Feyn,cont}(y,\vec{x}_{3})K_{l_{2}^{x}}^{Feyn,cont}(y,\vec{x}_{4}))
\end{equation}
\normalsize

If we take the inverse Fourier transform of Equation~\ref{eq:contribution}, it passes inside the integrals as the transform is not in terms of any spatial coefficients. Each propagator's inverse Fourier transform is the volume of directed paths between its endpoints (i.e. it is the continuum multinomial coefficients), and we know that the original expression involves convolutions across points adjoining propagators; by the relation between convolutions and products via the fourier transform we would obtain for our expression a product of the volumes of path spaces. This is what we would obtain, however, if we wanted to measure the volume of directed paths that arrange themselves to `look like' the Feynman diagram in Figure~\ref{fig:contrib}. That is, if we want to see the contribution to some correlator for this theory associated to the first power of $\lambda$, then that will be a sum over direct path-space volumes who agree with Feynman diagrams having $1$ split in them. This would show how some geometric model of the path integral can give rise to any perturbative QFT, should we be able to define spin $>.5$ propagators in this geometric setting. A more in-depth theory relating to this result is proven in Section~\ref{section:Interactions}.

\subsection{Statement of Theorems}

Theorem~\ref{thm:k3} is a extension of the work of \cite{odwyer2023relativistic} to $d>1$:

\begin{theorem}
\label{thm:k3}

Let $d\in\mathbb{N}$. Furthermore, say that the distribution of rational points on $S^{d}\in\mathbb{R}^{d+1}$ is equidistributed and that the asymptotic density of $\mathcal{A}_{n}$ as $t\rightarrow \infty$ is $\Theta(t^{d})$. Then, for $\vec{x},\vec{y}\in\mathbb{R}^{d}\times\mathbb{R}$ we know $K_{d_{n}}^{cont}(\vec{x},\vec{y})$ exists as does its limit over $n$. That limit becomes the function below.

$$K_{l_{2}^{*}}^{Feyn,cont}(\vec{x},\vec{y})=\mathcal{F}|_{m}^{I}((1-\frac{\sum_{i=1}^{d}x_{i}^{2}+I^{2}}{t^{2}})^{\frac{d-2}{2}})=Cm^{\frac{d-1}{2}}H_{\frac{d-1}{2}}^{(2)}(md_{l_{2}^{*}}(0,\vec{x}))$$

were $C$ is allowed to contain a complex phase. amd $\mathcal{F}$ denotes the Fourier transform.

\end{theorem}

As shown in \cite{Hong_Hao_2010}, this is the Klein Gordon propagator for a d-dimensional scalar particle. Now, we have a Theorem for 2d Manifolds embedded on the hyperboloid model of hyperbolic space.

\begin{theorem}
    \label{thm:hyperboliodModel}
    Let $\vec{x},\vec{y}\in X^{cont}$ with associated Fuschian group  $G$. We denote the flat space propagator on $\mathbb{Z}^{2}\times\mathbb{Z}$ as $K_{n}^{free,cont}(\vec{x},\vec{y})$. We have the following:

    $$K_{n}^{cont}(\vec{x},\vec{y})=\frac{\sum_{y'\in Orb_{G}(y),proj_{t}(y'-\vec{x})\le proj_{x}(y'-x)}K_{n}^{free,cont}(\vec{x},\vec{y}')}{|\{y'\in Orb(y)|proj_{x}(\vec{x}-\vec{y}')\le proj_{t}(\vec{x}-\vec{y}')\}|}$$

    and therefore
    \small
    $$K_{l_{2}}^{cont}(\vec{x},\vec{y})=\frac{\sum_{y'\in Orb_{G}(y),proj_{t}(y'-\vec{x})\le proj_{x}(y'-x)}\mathcal{F}|_{I}^{m}((1-\frac{proj_{x}(\vec{x}-\vec{y'})^{2}+proj_{y}(\vec{x}-\vec{y}')^{2}+I^{2}}{proj_{t}(\vec{x}-\vec{y}')^{2}})^{-\frac{1}{2}})}{|\{y'\in Orb(y)|proj_{x}(\vec{x}-\vec{y}')\le proj_{t}(\vec{x}-\vec{y}')\}|}$$
    \normalsize    
\end{theorem}

This yields a simple conjecture, which we will explain further in Section~\ref{sec:hyperboloid}, concerning the Kallen-Lehmann form \cite{Srednicki_2007} of the manifold propagator.

\begin{conjecture}

For some manifold M whose universal cover is $\mathbb{H}^{d}$ or $\mathbb{R}^{d}$ consider the set $\{\tau_{i}|y_{i}\in Orb(y)\}$ in the hyperboloid model (or pacman model from \cite{odwyer2023relativistic}. These deck transformations can have their proper times $\tau_{i}$ partitioned into disjoint periodic spacings collection that we denote by $\alpha_{i}\mathbb{N}+b_{i}$. This whole partition we denote by $W_{M}$.

Then free scalar propagator defined on manifolds whose universal cover is $\mathbb{H}^{d}$ or $\mathbb{R}^{d}$ have a Kallen-Lehmann power spectra $\rho(m)$ equivalent to $\delta(m)+\sum_{i=1}^{\infty}\frac{a_{i}\delta(m-m_{i})}{m}$ where $\{m_{i}\}_{i=1}^{\infty}\in \mathbb{R}$ is a discrete set whose only cluster point is zero but if $m_{i}\rightarrow 0$ then $a_{i}\rightarrow 0$. For each $m_{i}$ there is some $\alpha_{i}$ such that $\alpha_{i}m_{i}\in 4\pi \mathbb{N}$ where $\alpha_{i}$ is among the set $W_{M}$. This is a theorem for manifolds whose universal cover is $\mathbb{R}^{d}$.

\end{conjecture}

Now, we have a theorem concerning Tiechmuller curves. This could be applicable towards billiard problems. We present the theorem here: 

\begin{theorem}
    \label{thm:tiech}
    Let $\vec{x},\vec{y}\in\partial X,\vec{x}\neq\vec{y}$ for X a Tiechmuller curve (as defined in Section~\ref{section:def}). Let 



    $$\Gamma_{n,int,\vec{a}_{1},\vec{a}_{2}}^{\vec{x},\vec{y}}=\Gamma_{n}^{\vec{x},\vec{y}}\cap\{\gamma|\gamma\cap\partial X\subset\partial\gamma,\Delta\gamma_{i}=\vec{a}_{1},\Delta\gamma_{f}=\vec{a}_{2}\}$$

    where $\Delta \gamma_{i},\Delta\gamma_{f}$ denotes the slope of $\gamma$ at the initial and final points. If we put $0$ in either direction argument of $\Gamma^{\vec{x},\vec{y}}_{n,int,\vec{a}_{1},\vec{a}_{2}}$ we mean that the direction at that point is unconstrained. Let $K_{n,int,\vec{a}_{1},\vec{a}_{2}}^{feyn,cont}$ denote the path integral propagator taken over this latter collection of directed paths (and $K_{n,int}^{feyn,cont}$ denote that which is taken over the former). Then these propagators exist, $K_{n}^{Feyn,cont}(\vec{x},\vec{y})$ exists, and satisfies Equation~\ref{multline:teichmullerShit}
    \begin{multline}
    \small
    \label{multline:teichmullerShit}
    K_{n}^{Feyn,cont}(\vec{x},\vec{y})=K_{n,int,0,0}^{Feyn,cont}(\vec{x},\vec{y})\\
    +\sum_{M=1}^{\infty}\mathcal{F}|_{m}^{I}(\int...\int\int_{\partial X}...\int_{\partial X}\sum_{\vec{a}_{1},...,\vec{a}_{M-1}\in \mathcal{A}_{n}}\Pi_{j=1}^{j=M-1}\mathcal{F}^{-1}|_{m}^{I_{0}}(K_{n,int,\vec{a}_{j},\vec{a}_{j+1}}^{Feyn,cont}(w_{j},w_{j+1}))\\
    \mathcal{F}^{-1}|_{m}^{I_{j}}(K_{n,int,\vec{a}_{M-1},0}^{Feyn,cont}(\vec{x},\omega_{1}))\mathcal{F}^{-1}|_{m}^{I_{M-1}}(K_{n,int,\vec{a}_{M-1},0}^{Feyn,cont}(w_{M-1},\vec{y})\delta(I-\sum_{j=1}^{M-1}I_{j}))\Pi_{j=1}^{M-1}dw_{j}dI_{j})
    \normalsize
    \end{multline}  
\end{theorem}

Finally, we find the fermion propagator in $d=1$

\begin{theorem}
\label{thm:fermi}

Then for $\vec{x},\vec{y}\in\mathbb{Z}^{d}\times\mathbb{Z}$ we know $K_{n,\frac{1}{2}}(\vec{x},\vec{y})^{\mu,\nu}(\vec{x},\vec{y})$ exists and

$$K_{l_{2}^{*},\frac{1}{2}}^{Feyn,cont}(\vec{x},\vec{y})=(i\gamma_{\mu}\partial^{\mu}-m)\mathcal{F}|_{m}^{I}((1-\frac{\sum_{i=1}^{d}x_{i}^{2}+I^{2}}{t^{2}})^{\frac{d-2}{2}})$$
\end{theorem}

This, as it is the operator $\gamma_{u}p^{u}-m$ operating on the scalar propagator in \cite{Hong_Hao_2010}, is the Dirac propagator for a d-dimensional Dirac spinor particle. The theorems and conjectures related to a general model of QFT are posed in Section~\ref{section:Interactions}.

\section{Proof of the $d>1$ Flat Space Theorem}
\label{sec:flatspace}
First, let us obtain the proof of Theorem~\ref{thm:k3}. This is a direct extension of the work in \cite{odwyer2023relativistic}, and therefore I will only be including the minimum necessary details to said proof.

\begin{proof}

We let $\vec{x},\vec{y}\in X=\mathbb{Z}\times \mathbb{Z}$ such that $proj_{t}(\vec{y}-\vec{x})>0$. Let $\mathcal{A}_{n}$ be as derived in Theorem~\ref{thm:axes}. For $\gamma\in \Gamma_{n}$ we know by Theorem~\ref{thm:generate} that $\gamma$'s difference sequence may be drawn from $\mathcal{A}_{n}$; let $I_{\vec{a}}$ be the number of elements of the difference sequence of $\gamma$ equal to $\vec{a}$. Then, we note that for any path the following properties hold:

\begin{itemize}

\item $\mathcal{I}_{x_{\mu}}:\sum_{a\in\mathcal{A}_{n}}I_{a}proj_{x_{\mu}}(\vec{a})=proj_{x_{\mu}}(\vec{y}-\vec{x})$

\item $\mathcal{II}:I=\sum_{a\in\mathcal{A}_{n}}I_{a}d_{n}(0,\vec{a})$

\end{itemize}

Here, we are denoting $\rho_{n}(\gamma)$ by I. Using the freedom to swap difference sequences of some $\gamma$ as done in the analogous step in \cite{odwyer2023relativistic}, we have $K_{n}(\vec{y},\vec{x})$ takes on the form in Equation~\ref{eq:propStep1}.

\begin{equation}
	\label{eq:propStep1}
	K_{n}(\vec{y},\vec{x})=\sum_{I}(\sum_{\{I_{a}\}_{\vec{a}\in\mathcal{A}_{n}}\in\mathcal{I},\mathcal{II}}\begin{pmatrix}\sum_{\vec{a}\in\mathcal{A}_{n}}I_{a}\\\Pi_{a\in\mathcal{A}_{n}}(I_{\vec{a}})\end{pmatrix})e^{mI}
\end{equation}

Just as in \cite{odwyer2023relativistic}, any solution to $\mathcal{I}$ and $\mathcal{II}$ can be obtained by doing the following. We let $I_{(0,1)}=I-\sum_{\vec{a}\in\mathcal{A}_{n}\setminus\{(0,1)\}}I_{\vec{a}}d_{n}(0,\vec{a})$. We constraint $I_{(0,1)}$ because it's the element of $\mathcal{A}_{n}$ that necessarily doesn't have a position coordinate and can't be used to remove the other constraints in $\mathcal{I}$. For the other d+1 constraints, we first note that $\mathcal{A}_{n}$ contains at least $2d$ many simple null vectors $\pm \vec{e}_{i}+\hat{t}$ such that $d_{n}(0,\pm \vec{e}_{i}+\hat{t})=0$. We remove all remaining spatial constraints in $\mathcal{I}$ with these because they don't mix with condition $\mathcal{II}$, letting $I_{(e_{i},1)}=proj_{x_{i}}(\vec{y}-\vec{x})-\sum_{\vec{a}\in\mathcal{A}_{n}\setminus\{(e_{i},1)\}}I_{\vec{a}}proj_{x_{I}}(\vec{a})$. We note that the equations we obtain to remove all conditions so far constrained have the property that each variable only features in the equation that constrains it, meaning that we removed the implicit nature of condition $\mathcal{II}$ and the spatial parts of $\mathcal{I}$. Now, we can remove the last constraint in Equation~\ref{eq:lastone}
\small
\begin{multline}
\label{eq:lastone}
I_{-\vec{e}_{1}}=proj_{t}(\vec{y}-\vec{x})-\sum_{\vec{a}\in\mathcal{A}_{n}}I_{\vec{a}}proj_{t}(\vec{a})
\\=proj_{t}(\vec{y}-\vec{x})-\sum_{\vec{a}\in\mathcal{A}_{n}\setminus\{(0,1),(\vec{e}_{i},1)\}}proj_{t}(\vec{a})-\sum_{i=1}^{d}I_{(e_{i},1)}-I_{(0,1)}
\\
=proj_{t}(\vec{y}-\vec{x})-\sum_{\vec{a}\in\mathcal{A}_{n}\setminus\{(0,1),(\vec{e}_{i},1)\}}I_{\vec{a}}proj_{t}(\vec{a})-\sum_{i=1}^{d}(proj_{x_{i}}(\vec{y}-\vec{x})
\\-\sum_{\vec{a}\in\mathcal{A}_{n}\setminus\{(e_{i},1)\}}I_{\vec{a}}proj_{x_{i}}(\vec{a}))-(I-\sum_{\vec{a}\in\mathcal{A}_{n}\setminus\{(0,1)\}}I_{\vec{a}}d_{n}(0,\vec{a}))
\end{multline}
\normalsize

With these replacements, the discrete propagator of Equation~\ref{eq:propStep1} becomes Equation~\ref{eq:propStep2}.

\tiny
\begin{multline}
    \label{eq:propStep2}
    K_{n}(\vec{y},\vec{x})=\sum_{I}(\sum_{\Omega}\begin{pmatrix}\sum_{\vec{a}\in\mathcal{A}_{n}}I_{a}\\\Pi_{a\in\mathcal{A}_{n}\setminus\{(0,1),(e_{i},1),(-e_{1},1)\}}(I_{\vec{a}}),f_{x_{0}},f_{I},\Pi_{i=1}^{d}f_{x_{i}}\end{pmatrix})e^{mI}
    \\\textrm{ where }f_{x_{i}}=(proj_{x_{i}}(\vec{y}-\vec{x})-\sum_{\vec{a}\in\mathcal{A}_{n}\setminus\{(e_{i},1)\}}I_{\vec{a}}proj_{x_{I}}(\vec{a})),f_{I}=I-\sum_{\vec{a}\in\mathcal{A}_{n}}I_{\vec{a}}d_{n}(0,\vec{a})
    \\\textrm{, }f_{x_{0}}=proj_{t}(\vec{y}-\vec{x})-\sum_{\vec{a}\in\mathcal{A}_{n}\setminus\{(0,1),(\vec{e}_{i},1)\}}I_{\vec{a}}proj_{t}(\vec{a})-\sum_{i=1}^{d}(proj_{x_{i}}(\vec{y}-\vec{x})\\-\sum_{\vec{a}\in\mathcal{A}_{n}\setminus\{(e_{i},1)\}}I_{\vec{a}}proj_{x_{i}}(\vec{a}))-(I-\sum_{\vec{a}\in\mathcal{A}_{n}\setminus\{(0,1)\}}I_{\vec{a}}d_{n}(0,\vec{a}))
    \\
    \textrm{, }\sum_{\Omega}(*)=\sum_{I_{\vec{a}_{1}}=0}^{C_{1}}\textrm{ ... }\vec{a}\in\mathcal{A}\setminus\{(0,1),(e_{i},1),(-e_{1},1)\}\textrm{ ... }\sum_{I_{\vec{a}_{|\mathcal{A}_{n}|-d-2}}=0}^{C_{|\mathcal{A}_{k}|-d-2}}(*)
    \\\textrm{, and }C_{i}=min(\frac{proj_{t}(\vec{y}-\vec{x})-\sum_{j=1}^{i-1}I_{\vec{a}_{j}}proj_{t}(\vec{a}_{j})}{proj_{t}(\vec{a}_{i})},\frac{I-\sum_{j=1}^{i-1}I_{\vec{a}_{j}}d_{n}(0,\vec{a}_{j})}{proj_{t}d_{n}(0,\vec{a}_{i})}, l.c. \textrm{ constraint })
\end{multline}

\normalsize

The Equation~\ref{eq:propStep2} is intricate, so I will endeavor to explain each part. The multinomial coefficient term from Equation~\ref{eq:propStep1} has some of the $I_{\vec{a}}$ replaced with the expressions we derived before Equation~\ref{eq:propStep2}. The other $I_{\vec{a}}$ are completely unconstrained, other than the constraints placed upon them in the sum $\sum_{\Omega}(*)$. This is a simple sum over all possible values each $I_{\vec{a}_{i}}$ can take with the upper bound $C_{j}$ being determined by noting each step is monotonic in the time and phase components. Our expression in Equation~\ref{eq:propStep2} must necessarily have non-negative integer $f_{I}$ (as $I-\sum I_{\vec{a}}d_{n}(0,\vec{a})$ is granted to be non-negative from the bounds on $\sum_{\Omega}$) and $f_{x_{0}}$. So long as our steps never take us outside of the inverse light cone of $\vec{y}$, $f_{x_{i}}$ will also be positive. This is labelled as l.c. constraint in Equation~\ref{eq:propStep2} as is another simple linear upper bound like the other three in $C_{j}$.$\textrm{ }$With the discrete propagator specified, Theorem~\ref{thm:disctocont} allows us to immediately write the continuum version $K_{n}^{cont}$.

\tiny
\begin{multline}
    \label{eq:contProp}
    K_{n}^{cont}(\vec{y},\vec{x})=\mathcal{F}|_{I}^{m}(\int_{\Omega}\begin{Bmatrix}\sum_{\vec{a}\in\mathcal{A}_{n}}I_{a}\\\Pi_{a\in\mathcal{A}_{n}\setminus\{(0,1),(e_{i},1),(-e_{1},1)\}}(I_{\vec{a}}),f_{x_{0}},f_{I},\Pi_{i=1}^{d}f_{x_{i}}\end{Bmatrix})
    \\\textrm{ where }f_{x_{i}}=(proj_{x_{i}}(\vec{y}-\vec{x})-\sum_{\vec{a}\in\mathcal{A}_{n}\setminus\{(e_{i},1)\}}I_{\vec{a}}proj_{x_{I}}(\vec{a})),f_{I}=I-\sum_{\vec{a}\in\mathcal{A}_{n}}I_{\vec{a}}d_{n}(0,\vec{a})
    \\\textrm{, }f_{x_{0}}=proj_{t}(\vec{y}-\vec{x})-\sum_{\vec{a}\in\mathcal{A}_{n}\setminus\{(0,1),(\vec{e}_{i},1)\}}I_{\vec{a}}proj_{t}(\vec{a})-\sum_{i=1}^{d}(proj_{x_{i}}(\vec{y}-\vec{x})\\-\sum_{\vec{a}\in\mathcal{A}_{n}\setminus\{(e_{i},1)\}}I_{\vec{a}}proj_{x_{i}}(\vec{a}))-(I-\sum_{\vec{a}\in\mathcal{A}_{n}\setminus\{(0,1)\}}I_{\vec{a}}d_{n}(0,\vec{a}))
    \\
    \textrm{, }\int_{\Omega}(*)=\int_{I_{\vec{a}_{1}}=0}^{C_{1}}\textrm{ ... }\vec{a}\in\mathcal{A}\setminus\{(0,1),(e_{i},1),(-e_{1},1)\}\textrm{ ... }\int_{I_{\vec{a}_{|\mathcal{A}_{n}|-d-2}}=0}^{C_{|\mathcal{A}_{k}|-d-2}}(*)
    \\\textrm{, and }C_{i}=min(\frac{proj_{t}(\vec{y}-\vec{x})-\sum_{j=1}^{i-1}I_{\vec{a}_{j}}proj_{t}(\vec{a}_{j})}{proj_{t}(\vec{a}_{i})},\frac{I-\sum_{j=1}^{i-1}I_{\vec{a}_{j}}d_{n}(0,\vec{a}_{j})}{proj_{t}d_{n}(0,\vec{a}_{i})}, l.c. \textrm{ constraint })
\end{multline}

\normalsize

The next part of the proof from \cite{odwyer2023relativistic} involves showing that $K_{n}^{cont}$ is a Cauchy sequence in the $sup$ norm on functions. Having the spherical equidistribution and density results for pythagorean tupples, the exact same scaling arguments would work; $d_{l_{2}^{*}}$ is still absolutely continuous w.r.t. the Euclidean norm on time-like paths in higher dimensions. This fact, which is necessary to demonstrate the sequence of functions $K_{n}^{cont}$ is Cauchy, holds in this case as well. Therefore, we may assume $K_{l_{2}^{*}}$ exists. In the case of $d=1$ \cite{odwyer2023relativistic}, the final step after existence relied on demonstrating that the limiting object is invariant w.r.t rotations in $x$ and $I$ and boosts along $t$. For this purpose, we would work with a set of vectors $\mathcal{A}_{p}$ obeying a discrete rotational and hyperbolic symmetry that converged to a continuous in the limit of p and showed that the corresponding sequence of $K_{p}$ would converge to the same limit as that obtained from the pythagorean triples. This process is explained in detail in \cite{odwyer2023relativistic}; the only portion of the proof that do not generalize is the last step. In the last step, we treat the continuum multinomial coefficients as pdfs, and use Radon Nykodym theorem to assert that if the law of $\theta=sin(\frac{I}{\tau})$ is uniform then that of I is $(\tau^{2}-I^{2})^{-.5}$. In this case, the desired law of $I$ is $(\tau^{2}-I^{2})^{\frac{d-2}{2}}$, or the marginal distribution of points located on the $d+1$ sphere. That is, the distribution of $I$ is as though it were a single coordinate on a sphere of radius $t$ in $d+1$ dimensions; this is exactly the distribution you must have if the set $\{x_{i},I,t\}$ is invariant under rotations as obtained from the asymptotic rotational invariance of the pythagorean triples. This obtains the desired result.

\end{proof}

\section{Hyperboloid Model Path Integral}
\label{sec:hyperboloid}
The proof of Theorem~\ref{thm:hyperboliodModel} is an excercise in definitions; the hard part is Theorem~\ref{thm:k3} for $d=2$. Even so, let $X$ and $X^{cont}$ be $\mathbb{Z}\times \mathbb{Z}^{2}$ and $\mathbb{R}\times \mathbb{R}^{2}$ with equivalence relation $\sim$ as defined in the context of Theorem~\ref{thm:hyperboliodModel}. Then, we begin the proof.

\begin{proof}

First, lets write $K_{n}(\vec{x},\vec{y})=\sum_{\gamma\in\Gamma_{n}^{\vec{x},\vec{y}}}e^{im\rho_{n}(\gamma)}$. We note that a path from $\vec{x}\rightarrow \vec{y}$ is uniquely specified by what homotopy group elements it `enacts', or its final point upon lifting to the uniform hyperboloid cover and performing the same steps. This takes paths in $\mathbb{R}^{2}\times\mathbb{R}$ with our equivalence relation structure into those same paths $\mathbb{R}^{2}\times\mathbb{R}$ without said structure. Therefore $\Gamma_{n}^{\vec{x},\vec{y}}=\bigsqcup_{y'\in Orb(y)}\Gamma_{n}^{\vec{x},\vec{y}',free}$ and we have Equation~\ref{eq:funfun} and $\Gamma_{n}^{\vec{x},\vec{y}',free}$ are paths in $\mathbb{Z}^{2}\times \mathbb{Z}$ without the equivalence structure.

\begin{equation}
    \label{eq:funfun}
    K_{n}(\vec{x},\vec{y})=\sum_{y'\in Orb(y)}\sum_{\gamma\in\Gamma_{n}^{\vec{x},\vec{y'},free}}e^{im\rho_{n}(\gamma)}
\end{equation}

Now, the operator $\mathcal{T}_{m}^{cont}$, when applied to $K_{n}(\vec{x},\vec{y})$, will commute past the first sum. This is because $\mathcal{T}_{m}^{cont}$ only cares about the number of linear segments in a path in $\Gamma_{n}^{\vec{x},\vec{y}}$ and that would be independent w.r.t. what it translates into as you lift it into its corresponding path in $\mathbb{Z}\times\mathbb{Z}^{2}$. So, we obtain Equation~\ref{eq:contReim}

\begin{multline}
    \label{eq:contReim}
    K_{n}^{cont}(\vec{x},\vec{y})=lim_{m\rightarrow\infty}\frac{\mathcal{T}_{cont}^{m}K_{n}([m*\vec{x}],[m*\vec{y}])|_{m}}{\mathcal{T}_{cont}^{m}(max_{\vec{x'},\vec{y'}\in X|_{m},proj_{t}(\vec{y}'-\vec{x}')=[nt]}(\left|\Gamma_{n}^{\vec{x'},\vec{y'}}|_{m}\right|))}
    \\
    lim_{m\rightarrow\infty}\frac{\sum_{y'\in Orb(y)}\mathcal{T}_{cont}^{m}\sum_{\gamma\in\Gamma_{n}^{\vec{x},\vec{y'},free}}e^{im\rho_{n}(\gamma)}}{\mathcal{T}_{cont}^{m}(max_{\vec{x'},\vec{y'}\in X|_{m},proj_{t}(\vec{y}'-\vec{x}')=[nt]}(\left|\Gamma_{n}^{\vec{x'},\vec{y'}}|_{m}\right|))}
    \\
    \sum_{y'\in Orb(y)}lim_{m\rightarrow\infty}\frac{\mathcal{T}_{cont}^{m}\sum_{\gamma\in\Gamma_{n}^{\vec{x},\vec{y'},free}}e^{im\rho_{n}(\gamma)}}{\mathcal{T}_{cont}^{m}(max_{\vec{x'},\vec{y'}\in X|_{m},proj_{t}(\vec{y}'-\vec{x}')=[nt]}(\left|\Gamma_{n}^{\vec{x'},\vec{y'}}|_{m}\right|))}
\end{multline}

Now, we can consider multiplying and dividing by some factor related to the maximum paths in the free case, and obtain Equation~\ref{eq:contReim2}. 

\small
\begin{multline}
    \label{eq:contReim2}
    K_{n}^{cont}(\vec{x},\vec{y})=
    \sum_{y'\in Orb(y)}lim_{m\rightarrow\infty}\frac{\mathcal{T}_{cont}^{m}(max_{\vec{x'},\vec{y'}\in X|_{m},proj_{t}(\vec{y}'-\vec{x}')=[nt]}(\left|\Gamma_{n}^{\vec{x'},\vec{y'},free}|_{m}\right|))}{\mathcal{T}_{cont}^{m}(max_{\vec{x'},\vec{y'}\in X|_{m},proj_{t}(\vec{y}'-\vec{x}')=[nt]}(\left|\Gamma_{n}^{\vec{x'},\vec{y'}}|_{m}\right|))}
    \\
    \frac{\mathcal{T}_{cont}^{m}\sum_{\gamma\in\Gamma_{n}^{\vec{x},\vec{y'},free}}e^{im\rho_{n}(\gamma)}}{\mathcal{T}_{cont}^{m}(max_{\vec{x'},\vec{y'}\in X|_{m},proj_{t}(\vec{y}'-\vec{x}')=[nt]}(\left|\Gamma_{n}^{\vec{x'},\vec{y'},free}|_{m}\right|))}
\end{multline}
\normalsize

We note that the maximum number of paths for both the equivalence class and the free case grow at the same rate w.r.t $t$, and along with that that the denominator (equivalence class) maxima is a maxima over the sum of disjoint elements in the upper maxima. Since all the elements are positive (and the maxima achieved), we expect this ratio to converge to inverse of the number of deck lifts you have of the point $y'$, or $|Orb(y)|$, that have the property that they lie within the light cone of x. This obtains Equation~\ref{eq:contReim3}.

\small
\begin{equation}
    \label{eq:contReim3}
    K_{n}^{cont}(\vec{x},\vec{y})=
    \frac{\sum_{y'\in Orb(y)}K_{n}^{cont,free}(\vec{x},\vec{y}')}{|y'\in Orb(y)|proj_{x}(\vec{x}-\vec{y}')\le proj_{t}(\vec{x}-\vec{y}')|}
\end{equation}
\normalsize

Now, as in the proof of Theorem~\ref{thm:k3}, we include negative steps in $I$ to move to $K_{n}^{cont,feyn}$. This will not affect any of the steps previously mentioned. Furthermore, allowing more directions by increasing $m$ in the calculation of the continuous propagator (when thought of as a sum over path polytopes) does nothing to the discrete sum in Equation~\ref{eq:contReim3}. Therefore, the limit over m passes through this sum and applies directly to $K_{n}^{cont,free}$ in Equation~\ref{eq:contReim3}. Since Theorem~\ref{thm:k3} holds in $d=2$, this proves the desired result.

\end{proof}

Let's do a really simple non-compact Riemann surface example, the 3 branched cylinder. This entails finding an explicit equation for its deck transformations on some model of hyperbolic space and mapping it to the Hyperboloid (where our path integral is defined). We note any automorphism of the unit disk has the form $\frac{az+b}{\bar{b}z+\bar{a}}$ where $|a|^{2}-|b|^{2}=1$. The universal cover of the 3 branched cylinder is shown in Figure~\ref{fig:cover}

\begin{figure}[h]
    \centering
    \includegraphics[width = 6 cm]{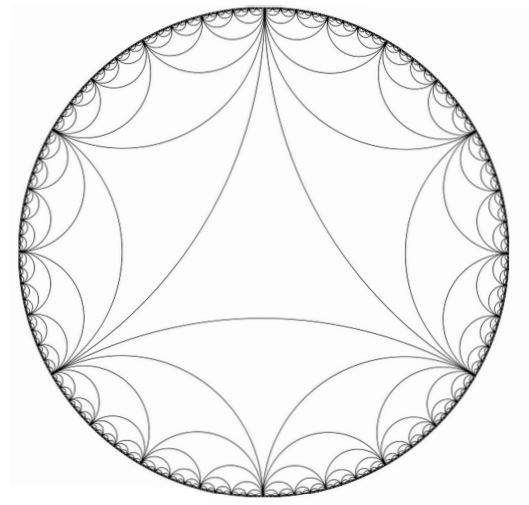}
    \caption{Universal Cover Mapping of the 3 Branched Cylinder}
    \label{fig:cover}
\end{figure}

We can note that we should be able to generate this tessellation with just three transformations; automorphisms of the unit disk that take the inner-most triangle into its adjacent 3 other triangles. We could then use compositions of these automorphisms to generate all of $\pi_{1}$ for the 3 branched cylinder. We note that the 3 points on the ends of this interior triangle are $e^{i(\{\frac{\pi}{2},\frac{\pi}{2}+\frac{2\pi}{3},\frac{\pi}{2}+\frac{4\pi}{3}\})}$. Let's first construct $\mathfrak{d}_{1}$, the deck transformation which takes the central triangle to its left neighbor. We know it fixes $e^{i(\frac{\pi}{2}+\frac{2\pi}{3})}$ and should take $e^{i\frac{\pi}{2}}\rightarrow e^{i(\frac{\pi}{2}+\frac{\pi}{3})}$ and $e^{i(\frac{\pi}{2}+\frac{4\pi}{3})}\rightarrow e^{i(\frac{\pi}{2})}$. This should completely constrain a and b, and indeed we obtain Equation~\ref{eq:firstdeck}.

\begin{equation}
    \label{eq:firstdeck}
    \mathfrak{d}_{1}(z)=-\frac{((2-i)+\sqrt{3})(i+2z)}{((2+i)+\sqrt{3})(2i+z)}
\end{equation}

$$$$

If we call $\mathfrak{d}_{2}(z)$ and $\mathfrak{d}_{3}(z)$ the right and up triangles, respectively, then it should be equal to $e^{-i\frac{2\pi}{3}}\mathfrak{d}_{1}(z)$ and $e^{-i\frac{4\pi}{3}}\mathfrak{d}_{1}(z)$ as their image is just rotated a bit. Composition by these three maps gives the entire homotopy group. The conversion between points on the Poincare Model and the hyperboloid model is well known, and if $z$ in the unit disk is written $re^{i\theta}$, then the conversion is in Equation~\ref{eq:conversion}.

\begin{multline}
    \label{eq:conversion}
    t=\frac{1+r^{2}}{1-r^{2}},x=\frac{rcos(\theta)}{1-r^{2}},y=\frac{rsin(\theta)}{1-r^{2}}
    \\
    r=\frac{\sqrt{x^{2}+y^{2}}}{(1+t)},\theta=arctan(\frac{x}{y}),
\end{multline}

This gives Equation~\ref{eq:value} for the value of $K_{l_{2}^{*}}(\vec{x},\vec{y})$ for $\vec{x}$ and $\vec{y}$ on the 3 branched cylinder:

\begin{equation}
    \label{eq:value}
    \frac{\sum_{y'=\mathfrak{d}_{i}^{(n_{i})}(\frac{proj_{x}(\vec{y})+iproj_{y}(\vec{y})}{1+proj_{t}(\vec{y})}),proj_{t}(\vec{y'}-\vec{x})\le proj_{x}(\vec{y'}-\vec{x})}\mathcal{F}|_{I}^{m}(\sqrt{1-\frac{\sum_{i\in\{x,y\}}proj_{i}(\vec{x}-\vec{y}')^{2}+I^{2}}{proj_{t}(\vec{x}-\vec{y}')^{2}}})}{|\{y'=\mathfrak{d}_{i}^{(n_{i})}(\frac{proj_{x}(\vec{y})+iproj_{y}(\vec{y})}{1+proj_{t}(\vec{y})}),proj_{t}(\vec{y'}-\vec{x})\le proj_{x}(\vec{y'}-\vec{x})\}|}    
\end{equation}

where $\vec{n}$ is a Smirnov word of arbitrary length composed from three letters and $n_{i}$ denotes letter frequency (this encodes arbitrary composition of $\mathfrak{d}_{i}$). This derivation demonstrates how the most difficult part in constructing the propagator for arbitrary 2 manifold is to find an explicit construction for the manifold's deck transformations; once this is complete, one may simply plug it into this discrete sum. One can envision that an alternate way to obtain the KG propagator on manifolds could then be used with this formula to probe deck transformations.

\subsection{Comments on the 1$\times$ 3 Case}


The results of Theorem~\ref{thm:k3} and Theorem~\ref{thm:k3},\ref{thm:hyperboliodModel} seem immediately extendable to $d=3$. If the pythagorean triples have the desired properties defined in the hypothesis of Theorem~\ref{thm:k3} for 5 tupples of pythagorean triples, then our procedure to generate the $K_{l_{2}}^{*}$ propagator in $\mathbb{R}^{2}$ and $\mathbb{H}$ extend to analogous procedures to define propagators for $\mathbb{R}^{3}\textrm{ and }\mathbb{H}^{3}$. These constant curvature spaces are clearly embeddable with minkowski signature in $4$ space (they would inherit Minkowski and de-Sitter structures as $4d$ manifolds). The same idea of covering $4$ space in slides of these three manifolds (hyperplanes and 1 sheet hyperboloids resp.)$\textrm{ }$would immediately extend Theorem~\ref{thm:hyperboliodModel}. Similarly, if we take the product of $\mathbb{H}^{2}$ embedding in $1\times 2$ space and $\mathbb{R}$ in itself, we get an analogous immersion of $\mathbb{H}^{2}\times\mathbb{R}$. Our work applies to these three Thurston Model geometries.


If we wanted to be exhaustive, then according to the geometrization theorem of 3 dimensional, compact manifolds \cite{defreitas2022geometrization}, we would need to define our procedure on 5 other simple, prime manifolds, and then define the propagator over the spherical cuts used to glue prime manifolds together. The procedure to glue manifolds together in geometrization is defined as the \textbf{connected sum}; one puts an equivalence relation on a copy of $\mathbb{S}^{d}$ inside each manifold thereby `glue-ing' them together. In Section~\ref{sec:tiech}, we introduce the method by which we would rigorously define the propagator over the gluing and that is by integrating over the boundary of the ball all tupples of path to a boundary point and then away from it (and possibly passing through the point multiple times). The author is convinced that any more rigorous treatment of gluing would be a tedious, but straightforward application of the work in Section~\ref{sec:tiech}. 

The Nil geometry and the universal cover of $SL(2,R)$ are described as `twisted' versions of $\mathbb{R}^{2}\times\mathbb{R}$ (i.e. something like a vector bundle over $\mathbb{R}$) and $\mathbb{H}^{2}\times\mathbb{R}$ respectively suggesting an extension to this setting may not be too difficult. Therefore, the difficult simply connected prime manifolds for which our definitions may not have an definition are $Sol,\mathbb{S}^{3}\textrm{ and }\mathbb{S}^{2}\times\mathbb{R}$. Sol is rather complicated, so the next domain to move to is $\mathbb{S}^{2}$. Despite it being the universal cover for only one manifold in $d=2$, defining it for this manifold could extend $\mathbb{S}^{3}$ and ultimately cover the definition for all but Sol. Then, every compact 3 manifold could have its propagator defined by our deck extension procedure on a simply connected 3 manifold and for simply connected 3 manifolds via stitching together our discrete paths between copies of the 8 simply connected prime manifolds. The author also would like to note that in the case of high dimensional manifolds, the continuum multinomial should be integrable across the surgery cuts used for their classification. This may yield the scalar path integral in a much more general regime.

\subsection{Kallen-Lehmann Form}
\label{subsec:klform}


The Kallen-Lehmenn form \cite{Srednicki_2007} of the QFT propagator takes the form of Equation~\ref{eq:funda}. It realizes the two point propagator of a complicated non-linear QFT as a sum over many different free particle propagators with different masses; any sharp prominence in $\rho(m)$ normally correspond to resonances or particles that arise naturally from the theory. It is a central question of QFT to demonstrate that certain 2 point correlators have no sharp prominence in $\rho(m)$ within some neighborhood of $m=0$ \cite{Srednicki_2007}.

\begin{equation}
\label{eq:funda}
\int dm \rho(m)H_{d}(x,m)=\pm\int dm \rho(m)H_{d}(0-m)=\pm\mathcal{F}|_{I}^{0}(\rho(I)(\tau^{2}-I^{2})^{\frac{d-2}{2}})
\end{equation}

In this section, the path integral on manifolds with $\mathbb{H}^{d}$ as their universal cover was demonstrated to be the discrete sum over flat space propagators originating from the origin to points in the image of deck transformations of our destination point. That is, the propagator (instead of being the marginal distribution of a single sphere centered at the origin) was the sum of marginal distributions from multiple spheres originating at the deck transformation images of the starting point. From simple multiplication this implies $\rho(I)=\frac{1}{N}\sum_{y'\in Orb(y)}(\frac{\tau_{i}^{2}-I^{2}}{\tau^{2}-I^{2}})^{\frac{d-2}{2}}rect(\frac{I}{\tau_{i}})$. Now in $d=2$ this simplifies to $\rho(I)=\frac{1}{N}\sum_{y'\in Orb(y)}rect(\frac{I}{\tau_{i}})$ and in that case $\rho(m)=\frac{1}{N}\sum_{y'\in Orb(y)}\tau_{i}sinc(\frac{\tau_{i}m}{2})$. It is the purpose of this section to gather evidence that this converges in large N limit as a distribution to the distribution $\delta(m)-\sum_{i=1}^{\infty}\frac{a_{i}\delta(m-m_{i})}{m}$ where $\{m_{i}\}_{i=1}^{\infty}$ are the discrete spectrum of this propagator and while $0$ is necessarily a cluster point of this set for all $m_{i}\rightarrow 0$, we have $a_{i}\rightarrow 0$ (this last condition is somewhat like the existence of a \textbf{mass gap} but unfortunately we seemingly dont have a true mass gap).

To obtain this result, consider for $m\in (-\infty,\infty)$ the set 

$$M_{N}=\{sin(\frac{\tau_{i}m}{2})|\tau_{i}\in Orb(y),proj_{t}(y)\ge N\}$$

Up to order $N$ the series $\sum_{m\in M_{N}}m$ is either $o(N)$ or $aN+o(N)$ for $a\in [-1,1]$ as the range of $sin$ is restricted to $[-1,1]$. In the first case let $\sum_{m\in M_{N}}m=f(n)+o(f(n))$, then the sequence of running averages converges (by the squeeze thoerem) to $0$ as it is bounded above and below asymptotically by $\pm 2f(N)N^{-1}\rightarrow 0$. The same reasoning asserts that the running average in the other case converges pointwise to $a$. 

We now need to see when $\frac{\tau_{i}m}{2}$ is uniform w.r.t $2\pi$; on this collection of $m$ our pointwise limit of $\rho(m)$ becomes zero. The nature of $\tau_{i}$, the proper time between the origin and deck transformations of the target point, suggests that there is some collection $\{(\alpha_{i},\beta_{i})\}_{i=1}^{N}$ such that $\tau'_{i}\in \cup_{i}\{\alpha_{i}\mathbb{N}+\beta_{i}\}$ (or that this is approximately true). Assume there is some subset of $\{\tau'_{i}\}=\alpha\mathbb{N}+\beta$ and that on this subset $\frac{\tau'_{i}m}{2}$ is equidistributed mod $2\pi$. This occurs iff (by Weyl's equidistribution theorem \cite{yifrach2023note}) $\alpha\in \mathbb{R}\setminus \frac{4\pi}{m}\mathbb{Q}$. Returning back to our set $\tau_{i}$, the set of $m$ which satisfies $\frac{\alpha_{i}m}{4\pi}\in\mathbb{Q}$ for some $\alpha_{i}$ is countable, and therefore so are the masses $m$ for which $\rho(m)$ converges to something non-zero.

Now, we may also consider those sets of $m$ such that $\frac{\alpha_{i}m}{4\pi}\in \mathbb{Q}$. While these aren't equidistributed, if we map them into $\mathbb{S}$ via the exponential map (so that the image under sin would be the imaginary portion of this mapping) they may also be radially symmetric on this circle. This implies their image under sin would destructively interfere; therefore the only $m$ which does not necessarily destructively decohere are those such that $\alpha_{i}m=4\pi \mathbb{N}$ (because these give a constant constructive value for $sin$. This gives us that the set $m_{i}$ where the above running average doesn't vanish is necessarily countable and discrete except possibly at $m=0$. The only way to ensure that $m=0$ is not a cluster point is if long $\alpha_{i}$ are regulated. However if $\alpha_{i}$ is some $\alpha$ at which we enjoy constructive interference, then so is any $N\alpha_{i}$. These higher modes will be farther spaced out in our series for $\rho(m)$ and regulated by $N$. Therefore, any sequence of masses $m_{i}\rightarrow m$ will have their $a_{i}$ reduced by this spacing and tend towards zero.

That this running average would converge to a sum of Dirac delta functions (regulated by $m^{-1}$) is obtained as it is the limit of distributions that concentrate non-zero 'mass' around a discrete and countable set (where $m^{-1}$ is obtained from the $m^{-1}$ in the $sinc$ function). If we multiply any test function $\psi\in\mathbb{C}^{\infty}$ times $\frac{1}{N}\sum_{y'\in Orb(b)}\tau_{i}sinc(\frac{\tau_{i}m}{2})$ and integrate, one may separate the integral into open sets about $\{\mathbb{B}_{N^{-.5}}(m_{i})\}\cup\{\mathbb{B}_{N^{-.5}}(0)\}$ and its complement. The integral of the complement is bounded by $N^{-1}$. Meanwhile its value at $m_{i}$ is approximately $\frac{a_{i}}{m}\phi(m_{i})$ where $a_{i}$ is independent of $m$ and test function. Here we note that in our definitions we should have divided by $\sqrt{N}$ rather than $N$, so that the integral at the point $m_{i}$ does not diminish as $N\rightarrow 0$.

Now, we state in what manner the periodicity of $\tau_{i}$ can be said to be true. In the pacman universal cover of $\mathbb{T}^{2}$ by $\mathbb{R}^{2}$, our deck transformations arrange themselves in a regular grid and the $\tau_{i}$ would be minkowski distances between regular points on a grid and some target point. Since the minkowski distance scales linearly with distance, it is immediate that the minkowski distances (like euclidean distances) for these points would occur along periodic distances from the target point. In manifolds with $\mathbb{H}^{2}$ as their cover, we note that the orbits of our target points will not be arranged along lattice points but along sheets of a hyperbola. However, as the target point moves farther out from the origin the local curvature of the hyperbola becomes asymptotically flat and lends credence to the idea that $\tau_{i}$ may still be arranged (approximately) in periodic orbits.



\section{Tiechmuller Curve Theorem}
\label{sec:tiech}
In this section, we prove existence and a nice relation for the propagators on Tiechmuller Curves $X$ given in Theorem~\ref{thm:tiech}. It should become obvious how one would define propagators over connected sum manifolds (as mentioned in Section~\ref{sec:hyperboloid}) because this proof will involve integrating the propagator over boundaries. 


\begin{proof}
Consider the decomposition in Equation~\ref{eq:decompo}. $\bigotimes$ denotes a cartesian product. This expression denotes the fact that if you have some $\gamma\in \Gamma_{n}^{\vec{x},\vec{y}}$ is must intersect (i.e. cross over) $\partial X$ some i number of times. We can therefore equate $\Gamma_{n}$ into a disjoint union over i of the set of paths that cross i times. If we then look at the set of paths that cross over $i$ times, that is equal to some cartesian $i$ dimensional product where the first entry is a path from you origin point to the boundary (hitting the boundary with some angle) you have $n-1$ paths between points on the boundary, and then a path from your boundary to the last point originating with some initial angle. Once we take the disjoint union over possible directions these junctures can take we get Equation~\ref{eq:decompo}.

\begin{equation}
    \label{eq:decompo}
    \Gamma_{n}^{\vec{x},\vec{y}}=\bigsqcup_{i=1}^{N}(\bigsqcup_{a_{1},...,a_{N-1}\in\mathcal{A}_{n}}\Gamma_{n,int,0,\vec{a}_{1}}^{\vec{x},\partial \gamma}\bigotimes_{j=1}^{N-1}\Gamma_{n,int,\vec{a}_{j},\vec{a}_{j+1}}^{\partial \gamma,\partial \gamma}\bigotimes\Gamma_{n,int,\vec{a}_{N-1},0}^{\partial \gamma,\vec{y}})
\end{equation}

In this equation, $N$ is some finite number. If $D=diam(X)$, then $N=\lceil\frac{D}{proj_{t}(\vec{y}-\vec{x})}\rceil$ is sufficient. We shall let it be this value for now. Using Equation~\ref{eq:decompo}, and specifically the disjoint union of the path sets, we obtain Equation~\ref{eq:tiechpath} for $K_{n}^{feyn}$.

\begin{multline}
    \label{eq:tiechpath}
    K_{n}^{feyn}(\vec{x},\vec{y})=\sum_{\gamma\in\Gamma_{n}^{\vec{x},\vec{y}}}e^{im\rho_{n}(\gamma)}=\sum_{i=1}^{N}\sum_{a_{1},...,a_{N-1}\in\mathcal{A}_{n}}
    \\\sum_{\gamma\in \Gamma_{n,int,0,\vec{a}_{1}}^{\vec{x},\partial \gamma}\bigotimes_{j=1}^{i-1}\Gamma_{n,int,\vec{a}_{j},\vec{a}_{j+1}}^{\partial \gamma,\partial \gamma}\bigotimes\Gamma_{n,int,\vec{a}_{N-1},0}^{\partial \gamma,\vec{y}}}e^{im\rho_{n}(\gamma)}
\end{multline}

Now, we want to demonstrate the existence of the continuum object upon applying $\mathcal{T}^{m}$ to Equation~\ref{eq:tiechpath} and normalizing it. One can see that this is a corollary of the work in Section~\ref{sec:flatspace} via the following argument. The directed paths that lie in $\Gamma_{n,int,0,0}^{\vec{x},\vec{y},cont}$ may still be partitioned into collections by the number of linear segments in said paths. These collections are still finite dimensional spaces, and each has less volume than it would have were it composed of any directed path as they are confined to not cross $\partial X$. Since we have demonstrated the sum of finite volumes of parameter spaces of $\Gamma_{n}^{\vec{x},\vec{y}}$ for $\vec{x},\vec{y}\in\partial\gamma$ is absolutely convergent in Section~\ref{sec:flatspace}, the same must be true of the sum of volumes of directed path spaces for our constrained case. Now consider $\Gamma_{n,int,\vec{a}_{1},\vec{a}_{2}}^{\vec{x},\vec{y},cont}$. The expression for the measure of $\Gamma_{n}^{\vec{x},\vec{y}}$ as a continuum multinomial is composed of a sum over dyke words which capture the order of directions one uses along a path. To obtain the measure of $\Gamma_{n,int,\vec{a}_{1},\vec{a}_{2}}^{\vec{x},\vec{y},cont}$, we only sum over the dyke words that begin with letter corresponding to $\vec{a}_{1}$ and end with that corresponding to $\vec{a}_{2}$. This measure, for the same reasons as $\Gamma_{n,\int,0,0}^{\vec{x},\vec{y},cont}$, will converge; it is a series that is composed of reordering and diminishing of terms of the series representing the measure of $\Gamma_{n}^{\vec{x},\vec{y}}$ and therefore will also converge absolutely. The propagators over these path spaces will converge as they represent changing the sign of terms in an absolutely convergent series. Finally, $K_{n}^{fen,cont}$ must exist as it will be a finite sum (the sum $\sum_{i=1}^{N}\sum_{a_{1},...,a_{N-1}\in\mathcal{A}_{n}}$ persists through the application of $\mathcal{A}_{n}$) of the convergent propagators over the continuum version of path spaces specified in Equation~\ref{eq:decompo}. Now, we show the nice relation for $K_{n}^{feyn,cont}$.

In Section~\ref{section:intro} we introduced the notation $\#_{B}^{I}$ to denote the measure of paths between two points as a function of their minkowski lengths $I$ and showed in Section~\ref{sec:flatspace} demonstrated that $K_{n}^{feyn,cont}$ can be rigorously treated as the fourier transform of $\#_{B}^{I}$ from $I$ to m. As $\#_{B}^{I}$ doesn't have enough indices to to capture the index $n$ from $\mathcal{A}_{n}$ as well as many other nuances to our path spaces in Equation~\ref{eq:decompo}, we use our work from Section~\ref{sec:flatspace} to denote the measure of $\Gamma_{n,int,\vec{a}_{N-1},0}^{\partial \gamma,\vec{y}}$ by $\mathcal{F}^{-1}|_{m}^{I}(K_{n,int,\vec{a}_{j},\vec{a}_{j+1}}^{feyn,cont}(\vec{x},\vec{y}))$. If we fix the points in $\partial \gamma$ in the expression $\Gamma_{n,int,0,\vec{a}_{1},cont}^{\vec{x},\partial \gamma}\bigotimes_{j=1}^{i-1}\Gamma_{n,int,\vec{a}_{j},\vec{a}_{j+1}}^{\partial \gamma,\partial \gamma,cont}\bigotimes\Gamma_{n,int,\vec{a}_{N-1},0}^{\partial \gamma,\vec{y},cont}$, then as a cartesian product is measure will be the product of measures of each individual part. Namely we have that is would be Equation~\ref{eq:firststepTiechmuller}

\begin{equation}
    \label{eq:firststepTiechmuller}
    \mathcal{F}^{-1}|_{m}^{I_{0}}(K_{n,int,0,\vec{a}_{1}}^{feyn,cont}(\vec{x},\partial\gamma))\Pi_{j=1}^{i-1}\mathcal{F}^{-1}(K_{n,int,\vec{a}_{j},\vec{a}_{j+1}}^{feyn,cont}(\partial \gamma,\partial \gamma))\mathcal{F}^{-1}|_{m}^{I_{N}}(K_{n,int,\vec{a}_{N-1},0}^{feyn,cont}(\partial \gamma,\vec{y}))
\end{equation}

Our propagator is a sum over every possible collection of our boundary points in the discrete setting; this becomes an integral under the limit by $\mathcal{T}_{cont}$ and we obtain the Equation~\ref{eq:fullTiechmuller} for the measure of the continuum paths that cross the boundary $N$ times.

\begin{multline}
    \label{eq:fullTiechmuller}
    \int_{\partial X}...\int_{\partial X}\mathcal{F}^{-1}|_{m}^{I_{0}}(K_{n,int,0,\vec{a}_{1}}^{feyn,cont}(\vec{x},\omega_{1}))\Pi_{j=1}^{i-1}\mathcal{F}^{-1}|_{m}^{I_{j}}(K_{n,int,\vec{a}_{j},\vec{a}_{j+1}}^{feyn,cont}(\omega_{j},\omega_{j+1}))
    \\
    \mathcal{F}^{-1}|_{m}^{I_{N}}(K_{n,int,\vec{a}_{N-1},0}^{feyn,cont}(\omega_{N-1},\vec{y}))\Pi_{j=1}^{N-1}d\omega_{j}
\end{multline}

We make the observation that we have many distinct paths with lengths $I_{1},...,I_{N-1}$ in Equation~\ref{eq:fullTiechmuller}. Ultimately we only want to collect paths that path through the boundary $N$ times and have the total length $I$. This introduces the further integral and $\delta$ we observe in Equation~\ref{eq:fullTiechmuller2}

\begin{multline}
    \small
    \label{eq:fullTiechmuller2}
    \int_{\mathbb{R}}...\int_{\mathbb{R}}\int_{\partial X}...\int_{\partial X}\mathcal{F}^{-1}|_{m}^{I_{0}}(K_{n,int,0,\vec{a}_{1}}^{feyn,cont}(\vec{x},\omega_{1}))\Pi_{j=1}^{i-1}\mathcal{F}^{-1}|_{m}^{I_{j}}(K_{n,int,\vec{a}_{j},\vec{a}_{j+1}}^{feyn,cont}(\omega_{j},\omega_{j+1}))
    \\
    \mathcal{F}^{-1}|_{m}^{I_{N}}(K_{n,int,\vec{a}_{N-1},0}^{feyn,cont}(\omega_{N-1},\vec{y}))\delta(I-\sum_{j=1}^{N-1}I_{j})\Pi_{h=1}^{N-1}d\omega_{j}dI_{j}
    \normalsize
\end{multline}

We pass the sum $\sum_{i=1}^{N}\sum_{a_{1},...,a_{N-1}\in\mathcal{A}_{n}}$ from Equation~\ref{eq:decompo} into this expression. Finally we take the fourier transfrom back from $I$ to $m$ to obtain the expression in Theorem~\ref{thm:tiech}

\end{proof}

\section{Fermion Path Integral}
\label{sec:fermi}

In this section we will prove Theorem~\ref{thm:fermi}. 

\begin{proof}

We begin by manipulating the Dirac propagator's analytic form until it may be realized as a fourier transform of some volume of directed paths. It is well-known that we have the following Equation~\ref{eq:fembos} in general relating the Dirac propagator to the bosonic one \cite{pesky}. Here, we use $\#^{b}_{I}$ and $\#^{f}_{I}$ as shorthand for the volume of paths with $I$ as their bosonic or fermionic action.

\begin{multline}
    \label{eq:fembos}
    \mathcal{F}|_{I}^{m}(\#^{f}_{I})=(i\gamma^{\mu}\partial_{\mu}+m)\mathcal{F}|_{I}^{m}(\#^{b}_{I})
    \\
    \mathcal{F}|_{I}^{m}(\#^{f}_{I})=(i\gamma^{\mu}\partial_{\mu}-i(im))\mathcal{F}|_{I}^{m}(\#^{b}_{I})
    \\
    \mathcal{F}|_{I}^{m}(\#^{f}_{I})=\mathcal{F}|_{I}^{m}((i\gamma^{\mu}\partial_{\mu}-i\partial_{I})\#^{b}_{I})
    \\
    \#^{f}_{I}=i(\gamma^{\mu}\partial_{\mu}-\partial_{I})\#^{b}_{I}
\end{multline}

Next, we note that from taking the inverse Fourier transform of the KG propagator (as well as my previous work) we would obtain $\#^{b}_{I}\sim (t^{2}-x^{2}-I^{2})^{\frac{d-2}{2}}$. This and Equation~\ref{eq:fembos} implies Equation~\ref{eq:fermi2}.

\begin{multline}
\#^{f}_{I}=(\frac{d-2}{2})i(\gamma^{\mu}x_{\mu}-\mathbbm{1}I)(t^{2}-x^{2}-I^{2})^{\frac{d-4}{2}}
\\=\frac{d-2}{2}i(\frac{\gamma^{\mu}x_{\mu}-\mathbbm{1}I}{\sqrt{t^{2}-x^{2}-I^{2}}})(t^{2}-x^{2}-I^{2})^{\frac{(d-1)-2}{2}}
\label{eq:fermi2}
\end{multline}

Let's consider the following. We will evaluate both sides of Equation~\ref{eq:fermi2} by Dirac spinors $\vec{v},\vec{w}\in\mathbbm{C}^{2}$ (where evaluation is performed by multiplying by $v^{\dagger}$ on the left and $\vec{w}$ on the right. Keep in mind the initial frame vector $\vec{V}=(v^{\dagger}\gamma^{\mu}w)\hat{e}_{\mu}-v^{\dagger}w\hat{I}$ as we now denote the unit vector in along $(\vec{e}_{\mu},I)$ as $\vec{e}_{r}$ (i.e. a sort of a radial vector). Then, Equation~\ref{eq:fermi2} remarkably becomes Equation~\ref{eq:fermi3}

\begin{multline}
\#^{f}_{I}=\frac{d-2}{2}i(\vec{V}*\vec{e}_{r})(t^{2}-x^{2}-I^{2})^{\frac{(d-1)-2}{2}}
\\
=C(cosh\eta_{\vec{V}})(t^{2}-x^{2}-I^{2})^{\frac{(d-1)-2}{2}}
\\
=C(e^{\eta_{V}}+e^{-\eta_{V}})(t^{2}-x^{2}-I^{2})^{\frac{(d-1)-2}{2}}
\label{eq:fermi3}
\end{multline}

where $\eta_{V}$ is the rapidity between the radial direction and $\vec{V}$ (and C is some constant dependent only on the magnitude of $\vec{v}$ and $\vec{w}$).

It is a regular exercise in special relativity to show that the relationship between inner products and cosine of an angle translates to minkowski inner products and the cosh of rapidity; refer to \cite{10.1119/1.11972} for further work on this. Namely, $(t^{2}-x^{2}-I^{2})^{\frac{(d-1)-2}{2}}$ is the marginal distribution of a spherical distribution, which was shown in \cite{odwyer2023relativistic} and for scalar particles to arise as volume of a rotationally uniform space of directed paths (except in this case we have one less dimension than usual). The last equation of Equation~\ref{eq:fermi3} was included as it is what we can massage to get the particle lagrangian for the fermion as was down in Equation~\ref{eq:propagators} (we want to get some volume of paths times a complex phase factor). I now reintroduce the Fourier term to further make this comparison and to obtain Equation~\ref{eq:fermi4}.

\begin{multline}
v^{\dagger}K^{\frac{1}{2}}(t,\vec{x})w=C\int_{-\infty}^{\infty}e^{imI}(e^{\eta_{V}}+e^{-\eta_{V}})(t^{2}-x^{2}-I^{2})^{\frac{(d-1)-2}{2}}dI
\\
=C\int_{-\infty}^{\infty}e^{i(mI\pm i\eta_{V})}(t^{2}-x^{2}-I^{2})^{\frac{(d-1)-2}{2}}dI
\label{eq:fermi4}
\end{multline}

It is at this point that we move from manipulating the Dirac propagator to equating it to a fourier transform of a volume of directed paths. We will do so by proving Equation~\ref{eq:fermiStep1}.


\begin{multline}
    \label{eq:fermiStep1}
    lim_{n\rightarrow\infty}(lim_{m\rightarrow\infty}\frac{\mathcal{T}_{cont}^{m}\sum_{\gamma\in\Gamma_{n}^{\vec{x},\vec{y}}}e^{i(\sum_{\vec{a}\in\gamma}(md_{l_{2}}^{*}(\vec{a})+\eta_{\vec{a}}+i\nu\theta_{\vec{a}}))}}{\mathcal{T}_{cont}^{m}(max_{\vec{x'},\vec{y'}\in X|_{m},proj_{t}(\vec{y}'-\vec{x}')=[nt]}(\left|\Gamma_{n}^{\vec{x'},\vec{y'}}|_{m}\right|))})
    \\=C\sum_{b\in\{\pm 1\}}\int_{-\infty}^{\infty}e^{i(mI+ib\eta_{V})}(t^{2}-x^{2}-I^{2})^{\frac{(d-1)-2}{2}}dI
\end{multline}

where the LHS is the continuum object defined as $K_{l_{2}}^{feyn,v,w}$ in Section~\ref{sec:fermi}.The first easy step is to use the Einstein summation formula \cite{10.1119/1.11972} which is rewritten in Equation~\ref{eq:rapidity}. Rearranging the terms with the Wigner angle $a$, we have

\begin{multline}
    \label{eq:fermiStepAngle}
 e^{\eta_{v+w}}+e^{-\eta_{v+w}}=.5(e^{\eta_{v}}+e^{-\eta_{v}})(e^{\eta_{w}}+e^{-\eta_{w}})+.25(e^{\eta_{v}}-e^{-\eta_{v}})(e^{\eta_{w}}-e^{-\eta_{w}})(e^{ia}+e^{-ia})
\\=.25(e^{\eta_{v}}+e^{-\eta_{v}})(e^{\eta_{w}}+e^{-\eta_{w}}+e^{\eta_{w}}+e^{-\eta_{w}}+e^{\eta_{w}+ia}+e^{-\eta_{w}+i(a+\pi)}+e^{\eta_{w}-ia}+e^{-\eta_{w}-ia+i\pi})
\end{multline}

By our convention, each additional $\vec{a}$ in some path from our origin to a target point corresponds to a different path in our collection of sums with one of the above exponents $(e^{\eta_{w}},e^{-\eta_{w}},\textrm{ etc.})$. This means each step recombines so that at any point along the path our total sums of exponent is eqaul to cosh of the rapidity to that point along the path. This, and Equation~\ref{eq:propStep1} from Section~\ref{sec:flatspace}, obtains for us Equation~\ref{eq:fermiStep2}.

\begin{multline}
    \label{eq:fermiStep2}
    lim_{m\rightarrow\infty}\frac{\mathcal{T}_{cont}^{m}\sum_{\gamma\in\Gamma_{n}^{\vec{x},\vec{y}}}e^{i(\sum_{\vec{a}\in\gamma}(md_{l_{2}}^{*}(\vec{a})+\eta_{\vec{a}}+i\nu\theta_{\vec{a}}))}}{\mathcal{T}_{cont}^{m}(max_{\vec{x'},\vec{y'}\in X|_{m},proj_{t}(\vec{y}'-\vec{x}')=[nt]}(\left|\Gamma_{n}^{\vec{x'},\vec{y'}}|_{m}\right|))}\\=lim_{m\rightarrow\infty}C\mathcal{T}_{cont}^{m}\sum_{I}(\sum_{\{I_{a}\}_{\vec{a}\in\mathcal{A}_{n}}\in\mathcal{I},\mathcal{II}}\begin{pmatrix}\sum_{\vec{a}\in\mathcal{A}_{n}}I_{a}\\\Pi_{a\in\mathcal{A}_{n}}(I_{\vec{a}})\end{pmatrix})e^{imI+\eta_{\vec{V}}}\\
\end{multline}

Here, we have grouped together those paths with length $I$, giving us a single factor in the argument of the exponential with C used to capture our normalizing factor. The continuum limit of this expression becomes a carbon copy of Equation~\ref{eq:contProp} with a factor $e^{\eta_{\vec{V}}}$ due to the work of Section~\ref{sec:flatspace}. We now acknowledge that the alteration we introduced in the scalar proof (that is,$\textrm{ }$moving from $K^{cont}$ to $K^{cont,Feyn}$) of allowing each step to take $\pm I$ increments would imply that our expression would be necessarily even in $\eta_{V}$ immediately giving as our result Equation~\ref{eq:contProp} but times $cosh(\eta_{\vec{V}})$. At this point, it's almost definitional. The same proof mentioned in Section~\ref{sec:flatspace} and done in its full extent in \cite{odwyer2023relativistic} for the existence of the lattice scalar propagator works for the lattice Dirac propagator; showing or sequence of $K_{\frac{1}{2},n}^{cont,Feyn}$ is Cauchy and has some limit just requires the same real analysis steps.$\textrm{ }$The only thing we must demonstrate is that our symmetry arguments for the scalar propagator actually yield one less dimension in this instance (we don't get the marginal distribution of the $d$ sphere but the $d-1$ sphere). Just as in the scalar propagator argument, we would fix the final distribution using symmetry arguments. Every lorentz transformation, just like every rotation, is indexed by the choice of two vectors in an orthonormal basis of $\mathbb{R}^{d+1}$ (you can think of its as a rotation/boost about a plane in higher dimensions). If we choose one of the axes to be $\vec{V}$ (and make a basis after), then our transformations are invariant every transformation using a pairs of two vectors (none of which is $\vec{V}$). This means our resulting distribution $(t,\vec{x},I)$ will have be distributed as $(t^{2}-x^{2}-I^{2})^{\frac{(d-1)-2}{2}}$ when projected on the hyperspace orthogonal to $\vec{V}$.

The decomposition of $(t,\vec{x},I)$ into its projection along $\vec{V}$ and the hyperplane orthogonal is unique. Therefore, using the additivity property of probability ($P(x,y)=\sum_{y\in\Omega}P(x|y)P(y)$), we know that the probability that $(t,\vec{x},I)$ takes some value is the likelihood that we have a given hyperplane projection (which is $(t^{2}-x^{2}-I^{2})^{\frac{(d-1)-2}{2}}$) times the probability we then draw the some projection along $\vec{V}$. But, we already know that the probability of drawing vectors along $\vec{V}$ in terms of the likelihood of drawing some $\eta_{V}$, and therefore we simply multiply the pdf of $\eta_{V}$ to this original expression, and we get the law of $(t,\vec{x},I)$. That law is $\#_{I}^{f}$, and we are done.
\end{proof}

\subsection{Fermionic Path Integral On Curved Spaces}

The work of Section~\ref{sec:hyperboloid} extends immediately to the fermionic case, but we will not elaborate further on this. There ought to be a relationship between the power spectrum of the fermion and purely geometric quantities ala Subsection~\ref{subsec:klform} as well; in this section we will demonstrate something unique to the fermion: rapidity based cancellation. Let $d=3$. Then, our fermionic propagator in Euqation~\ref{eq:fermiStep1} simplifies to Equation~\ref{eq:berryphase}

\begin{equation}
    \label{eq:berryphase}
    \sum_{b\in \{\pm 1\}}\int e^{i(mI+ib\eta_{V})}rect(\frac{I}{\tau})dI
\end{equation}

In the curved space setting, this would become (without normalization) the Equation~\ref{eq:berryphase2}.

\begin{equation}
    \label{eq:berryphase2}
    \sum_{y'\in Orb(y)}\sum_{b\in \{\pm 1\}}\int e^{i(mI+ib\eta_{V}')}rect(\frac{I}{\tau'})dI
\end{equation}

where $\tau'$ and $\eta'$ are defined between our origin and the deck transformation of our target point. The addition of these rect which lead to a discrete power spectrum at geometrically defined mass points in Subsection~\ref{subsec:klform} is now complicated by the presence of a hyperbolic cosh which depends not only on the spacing of the deck transformation points but also their orientation w.r.t each-other. Therefore, while it is still quite likely that there is a geometric expression of the Kallen Lehmann power spectrum for the fermionic path integral on curved spaces, it will encode both information about the relative spacing of the deck transformations and their orientation w.r.t each-other (and is not easy to calculate).

\section{Interactions}
\label{section:Interactions}

In this section, we demonstrate how perturbative scalar QFT can be encoded in a geometric formalism.$\textrm{ }$We will be careful, in the statement of our final theorem, to mention that this formalism makes no promises about the convergence of perturbative scalar QFT; we simply assert that should our geometric model converge it would converge to the value expect in QFT as it obeys Feynman diagram rules. For the same to be said about QFT in general, we would need an analogous geometric picture for Spin 1 particles. It may be the case that said picture could be achieved by light cone paths (photons), where massive Spin 1 particles are really mass-less particles with a Higgs scalar coupling as realized by Stueckelberg. We will give this as a conjecture in the end of this section. Let us prove our foundational theorem relating the geometric interpretation to perturbative QFT. For this purpose, let us restate the Feynman rules. 


If our field Lagrangian has a three point term $\lambda \phi^{3}$, then we can consider the three `first order' corrections to our four point correlator in Figure~\ref{fig:contrib2}. These are not first order in the sense of number of loops but rather in terms of $\lambda$, the correlation strength.

\begin{figure}[h]
    \centering
    \includegraphics[width = 8 cm]{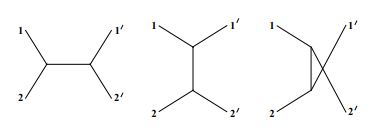}
    \caption{The 3 tree Level Contribution to the Four Field Correlator in $\lambda^{3}$ Theory \cite{Srednicki_2007}}
    \label{fig:contrib2}
\end{figure}
Focusing on the right-most term in Figure~\ref{fig:contrib2}, we can follow the following rules to obtain its contribution to the four point correlator. Let our contribution be $1$ from the start. We read the figure from left to right, ignoring the exterior arms. If we encounter a vertex of some degree, we multiply our contribution by $(i\lambda)$ where $\lambda$ is the coefficient associated with $\phi$ to that same degree in the action of the lagrangian. We label these vertices with points, and as we read our diagram left to right, we include a propagator $-iK$ between every labeled point that has a line in the diagram. Finally, we integrate over the points in the vertices, hence the $\int d^{d}y$ terms in Equation~\ref{eq:contribution}. This amounts to a full description of the Feynman rules for scalar QFT; a more eloquent explanation is given in \cite{Srednicki_2007},\cite{pesky}.

Consider taking the inverse Fourier transform of Equation~\ref{eq:contribution}. We would obtain Equation~\ref{eq:firstStep}.

\tiny
\begin{multline}
    \label{eq:firstStep}
    (i\lambda)^{2}(\frac{1}{i})^{5}\int d^{d}y d^{d}z K_{l_{2}^{*}}^{Feyn,cont}(y,z)(K_{l_{2}^{x}}^{Feyn,cont}(\vec{x}_{1},y)K_{l_{2}^{x}}^{Feyn,cont}(\vec{x}_{2},y)K_{l_{2}^{x}}^{Feyn,cont}(z,\vec{x}_{3})K_{l_{2}^{x}}^{Feyn,cont}(z,\vec{x}_{4}))
    \\
    =(i\lambda)^{2}(\frac{1}{i})^{5}\int d^{d}y d^{d}z \int_{-\infty}^{\infty}e^{iI_{0}m}(t^{2}-(x-y)^{2}-I^{2})^{\frac{d-2}{2}}dI_{0}\int_{-\infty}^{\infty}e^{iI_{1}m}(t^{2}-(\vec{x}_{1}-y)^{2}-I_{1}^{2})^{\frac{d-2}{2}}dI_{1}
    \\
    \int_{-\infty}^{\infty}e^{iI_{2}m}(t^{2}-(\vec{x}_{2}-y)^{2}-I_{2}^{2})^{\frac{d-2}{2}}dI_{2} \int_{-\infty}^{\infty}e^{iI_{3}m}(t^{2}-(\vec{x}_{3}-z)^{2}-I_{3}^{2})^{\frac{d-2}{2}}dI_{3} \int_{-\infty}^{\infty}e^{iI_{4}m}(t^{2}-(\vec{x}_{4}-z)^{2}-I_{4}^{2})^{\frac{d-2}{2}}dI_{4}
    \\
    =(i\lambda)^{2}(\frac{1}{i})^{5}\int d^{d}y d^{d}z\int\textrm{...}\int e^{im(I_{0}+I_{1}+I_{2}+I_{3}+I_{4})}\Pi_{j=0}^{4}(t^{2}-(X_{j})^{2}-I_{j}^{2})^{\frac{d-2}{2}}d I_{j}
    \\
    =(i\lambda)^{2}(\frac{1}{i})^{5}\int\textrm{...}\int e^{im(I_{0}+I_{1}+I_{2}+I_{3}+I_{4})}(\int d^{d}y d^{d}z\Pi_{j=0}^{4}(t^{2}-(X_{j})^{2}-I_{j}^{2})^{\frac{d-2}{2}}d I_{j})
    \\
    \textrm{ where }X_{0}=x-y,X_{1}=\vec{X}_{1}-y,X_{2}=\vec{x}_{2}-y,X_{3}=\vec{x}_{3}-z,\textrm{ and }X_{4}=\vec{x}_{4}-z.
\end{multline}
\normalsize

After the steps in Equation~\ref{eq:firstStep}, we consider taking the inverse Fourier transform $\mathcal{F}^{-1}|_{m}^{I}$ to obtain Equation~\ref{eq:secStep}.

\tiny
\begin{multline}
    \label{eq:secStep}
    \mathcal{F}^{-1}|_{m}^{I}(i\lambda)^{2}(\frac{1}{i})^{5}\int\textrm{...}\int e^{im(I_{0}+I_{1}+I_{2}+I_{3}+I_{4})}(\int d^{d}y d^{d}z\Pi_{j=0}^{4}(t^{2}-(X_{j})^{2}-I_{j}^{2})^{\frac{d-2}{2}}d I_{j})
    \\
    =(i\lambda)^{2}(\frac{1}{i})^{5}\int\textrm{...}\int \int dm e^{im(I_{0}+I_{1}+I_{2}+I_{3}+I_{4}-I)}(\int d^{d}y d^{d}z\Pi_{j=0}^{4}(t^{2}-(X_{j})^{2}-I_{j}^{2})^{\frac{d-2}{2}}d I_{j})
\end{multline}
\normalsize

Now, we are working in the setting (Schwartz class functions) where the Dirac delta function is defined as an operator on functions underneath an integral. As an operator, we have $\delta(x-y)=\frac{1}{2\pi}\int_{-\infty}^{\infty}dm e^{im(x-y)}$, so if we make this substitution into Equation~\ref{eq:secStep} we obtain Equation~\ref{eq:thirdStep}

\tiny
\begin{equation}
    \label{eq:thirdStep}
    (i\lambda)^{2}(\frac{1}{i})^{5}\int\textrm{...}\int (\int d^{d}y d^{d}z\Pi_{j=0}^{4}(t^{2}-(X_{j})^{2}-I_{j}^{2})^{\frac{d-2}{2}}\delta(I-\sum_{i=0}^{4}I_{i})d I_{j})
\end{equation}
\normalsize

Now we want to incorporate our lattice propagators into Equation~\ref{eq:thirdStep}, and we have shown $(t^{2}-(X_{j})^{2}-I_{j})^{2}$ is a point-wise limit of a multinomial coefficient representing the number of paths from 0 to $X_{j}$ aligned with the axes in $\mathcal{A}_{n}$. Using theorems about limits (the commutation of a limit with a product), we obtain Equation~\ref{eq:fourthStep}.

\tiny
\begin{multline}
    \label{eq:fourthStep}
    lim_{n\rightarrow\infty}(i\lambda)^{2}(\frac{1}{i})^{5}\int\textrm{...}\int (\int d^{d}y d^{d}z\Pi_{j=0}^{4}\int_{\Omega}\begin{Bmatrix}\sum_{\vec{a}\in\mathcal{A}_{n}}J_{a}\\\Pi_{a\in\mathcal{A}_{n}\setminus\{(0,1),(e_{i},1),(-e_{1},1)\}}(J_{\vec{a}}),f_{x_{0}},f_{I_{i}},\Pi_{i=1}^{d}f_{x_{i}}\end{Bmatrix}d\Omega\delta(I-\sum_{i=0}^{4}I_{i})d I_{j})
    \\
    \textrm{ where }f_{\pm}\textrm{ and }f_{x_{0}}\textrm{ are defined as in Equation~\ref{eq:propStep2}}
\end{multline}
\normalsize

We may interpret Equation~\ref{eq:fourthStep} more easily if we consider it from the view as the limit itself of a discrete object, or a number of discrete objects. We note that the continuous multinomial coefficient with the $\Omega$ integral is the limit of $|\Gamma|$ under $\mathcal{T}_{cont}$ for paths from $0$ to $X_{j}$. Under this limit, the integral over $y$ and $z$ arises from sums over the intermediate points $y,z\in\mathbb{Z}^{d}$. The multinomial coefficients tell us the number of paths between intermediate points, and the sum over locations of the intermediate points demonstrate that if we took the limit under $\mathcal{T}_{cont}$ of the number of lattice paths (with possibly multiple individual adjoining or splitting paths) arranged to look realize the left diagram in Figure~\ref{fig:contrib2}, then we would obtain the contribution of that Feynman diagram to the interaction four point correlator with a $\psi^{3}$ interaction term. This is done in the exact same way as Theorem~\ref{thm:properties} was obtained in \cite{odwyer2023relativistic}, so that,with a $\psi^{3}$ term the propagator, it can be understood as an expansion in $\lambda$ where the contribution for a given power of $\lambda$ is a limit over the number of paths between the initial and final points with that number of intersection points. This leads to Theorem~\ref{thm:QFT}.

\begin{theorem}
    \label{thm:QFT}
    Let $\mathcal{L}=(\nabla\phi)^{2}+\sum_{i=0}^{n}a_{i}\phi^{i}$ be the Lagrangian density of a scalar QFT in $d+1$ dimensions. Say the pythagorean tupples in $d+1$ are asymptotically rotationally symmetric. Then, should they be convergent, the $p$ point correlators $\mathbb{E}[T\phi(x_{1})...\phi(x_{p})]$ from perturbation theory are equivalent to taking the rigorous limit $\mathcal{T}_{m}$ of $\sum_{\sum n_{j}=1}^{\infty}(\Pi_{j=1}^{m}a_{j}^{n_{j}}\Gamma_{n}(\{n_{j}\})$ where $\Gamma_{n}(\{n_{j}\})$ denotes the number of lattice embedding of Feynman diagrams connecting our p points with $n_{j}$ many intersection points of degree $j$ with steps in $\mathcal{A}_{n}$. We obtain a continuum object representing the volume of the space of directed continuous lattice embedding of Feynman diagrams in $\mathbb{R}^{d}$ with continuum steps in $\mathcal{A}_{n}$ and total length $I$. We then take the limit as $n\rightarrow \infty$, allowing more directions to be included in our measurement of path space. Finally, we take the Fourier transform of the resulting set of directed paths from $I$ to m.
\end{theorem}

This should give a geometric version of QFT for all scalar theories. One interesting consequence of this interpretation is that (should this $\mathcal{T}_{cont}$ operation be convergent) we evade the divergences normally associated with computing point correlators in momentum space. The earlier discussion with the fermionic propagator will, in the same exact manner, lead to a geometric interpretation of all QFT theories involving only scalar and Dirac spinor fields. The author believes a geometric interpretation ought to exist for all perturbative QFT. In the next subsection, he presents the evidence for a conjecture about for theories with a spin 1 particle. 

\subsection{Massive Spin 1 Particles}

First, we make a new definition: the $d_{n}$ \textbf{photon propagator}, as $K_{n,photon}(\vec{x},\vec{y})^{\mu\nu}=i\frac{\delta(d_{n}(0,\vec{y}-\vec{x}))}{\mu(\{\vec{x}|d_{n}(0,\vec{y}-\vec{x})=0\})}g^{\mu\nu}$ where the Lebesgue measure of the $d_{n}$ light cone is well-defined because it will be a polyhedral fan.  We note that this distribution has a distribution-wise limit as $n\rightarrow\infty$ which is simply Equation~\ref{eq:photon}.

\begin{multline}
    K_{photon,l_{2}}(\vec{x},\vec{y})^{\mu\nu}=\frac{i\delta(proj_{t}(\vec{x}-\vec{y})^{2}-\sum_{i}proj_{x_{i}}(\vec{x}-\vec{y})^{2})}{\mu(\{\vec{x}|d_{l_{2}}(0,\vec{y}-\vec{x})=0\})}g^{\mu\nu}\\=\frac{i\delta(proj_{t}(\vec{x}-\vec{y})^{2}-\sum_{i}proj_{x_{i}}(\vec{x}-\vec{y})^{2})}{Vol(\mathbb{S}^{d-1}(proj_{t}(y)-proj_{t}(x)))}g^{\mu\nu}
    \label{eq:photon}
\end{multline}

One thing to note is that $K_{n,photon}$ as of yet does not have a necessary notion of gauge associated with it. It is what you obtain when you fix the gauge of the photon propagator in the Feynman gauge (the mass-less Spin-1 particle gauge). If you want to reintroduce the gauge of this Spin-1 boson, you would apply the operator $g^{\mu,\nu}+(1-\frac{1}{\lambda})\frac{p^{\mu}p^{\nu}}{p^{2}}$ to our original light cone propagator. The normal interpretation of momentum as differential operators is a little difficult in this setting; one thing we can note is that $(1-\frac{1}{\lambda})\frac{p^{\mu}p^{\nu}}{p^{2}}$ in position space would look like the differential operator $(1-\frac{1}{\lambda})\partial_{\mu}\partial_{\nu}$ times the position space transform of $p^{-2}$ which we know is $\delta(\tau^{2})$. This means that the full position space contribution from the gauge terms would be $((1-\frac{1}{\lambda})\partial_{\mu}\partial_{\nu}\delta(\tau))*\delta(\tau)$ where $*$ denotes convolution. The resulting contribution to the propagator is still only supported on the light cone and will likely enforce some jump in derivatives along paths at junctures with the photon propagator.

With a geometric notion of the zero mass spin 1 propagator defined, we can move to the massive one. With the Higgs mechanism \cite{Srednicki_2007}, we consider a mass-less spin 1 field coupled to a scalar field with non-zero expectation. Expectation values of $\phi$ give rise to a non-zero mass for the otherwise mass-less spin 1 boson. Consider the $U(1)$ higgs-coupled lagrangian $\mathcal{L}=-\frac{1}{4}F^{\mu\nu}F_{\mu\nu}+|(\partial-iqA)\phi|^{2}-\lambda(|\phi|^{2}-\Phi^{2})^{2}$. The Feynman rules for correlators in this theory include your regular photon and scalar terms (which $\phi's$ mass coupling being $\lambda(1-2\Phi)$), a constant $\Phi^{4}$ which we will normalize out, and a four vertex $\phi^{4}$ term times $\lambda$. Noting that the Abelian Higgs Model is perturbative in some parameter, this necessitates the following theorem for the correlators of the Abelian Higgs Model.

\begin{conjecture}
Let $\{A^{\mu},\phi\}$ be the particle content of the Abelian Higgs Model. Then,  $\mathbb{E}[A^{\mu}(x_{1})\phi(x_{2})..\phi(x_{n-1})A^{\nu}(x_{n})]$, an arbitrary point correlator of this model, may be first obtained by the application of $\mathcal{T}_{cont}$ to $\sum_{\sum n_{1}+n_{2}=1}^{\infty}q^{n_{1}}\lambda^{n_{2}}\Gamma_{n}(\{n_{1},n_{2}\})$ where $\Gamma_{n}$ denotes the number of lattice embedding of Feynman diagrams connecting the photon and scalar points that contain $n_{1}$ many interaction terms where a photon path terminates into a $\phi$ and $n_{2}$ many terms wherein $\phi$ has a four point vertex where each direction is constrained to lying along $\mathcal{A}_{n}$ for the $\phi$ paths and that intersected with $\{\vec{a},d_{n}(0,\vec{a})=0\}$ for the photon paths. Note, as previously, that we will make this continuum object a function of $I$ which is the accumulated length of the scalar paths in minkowski distance. We then take the limit $n\rightarrow\infty$ and take the Fourier transform from $I$ to $\lambda(1-2\Phi)$. 
\end{conjecture}
%




\section{Conclusion}

With the work thus far stated, we capture much of the dynamics required for general QFT. We generalized the work of a previous paper, which showed how the $1-1$ scalar propagator can be realized as the Fourier transform of a genuine sum over paths with their length being the transform variable, to the case of $1-d$ dimensions. We were then able to show how this definition easily generalized to all $1-2$ manifolds with a spatial structure that admitted the hyperbolic plane as its universal cover, and how we could use the uniformization theorem in $3$ dimensions along with the path integral across a stitched boundary of manifolds to define our object on a general class of $1-3$ manifolds.

We then moved onto the Dirac propagator in anticipation of a general statement of interactions. We were able to show how a small modification of our work would produce the Dirac fermion propagator in $1-1$ dimension; the author fully anticipates this generalizing to higher dimensions. The author then inspected the rules of perturbative QFT. Resulting from a simplistic proof relating multiplication and convolution through the Fourier transform, the author demonstrated that the Feynman rules that are used to calculate point correlators in QFT arise naturally from the assumption that these point correlators are obtained from a Fourier transform over the number of paths of length I which realize a Feynman diagram in your manifold. This gives rise to a purely geometric interpretation of perturbative QFT which the author conjectures generalizes to all theories including Spin $\frac{1}{2}$ and Spin 1 particles.

In future work, the author aims to obtain a theory applying $\mathcal{T}_{cont}$ to higher dimensional simplicial complexes in space. The simplex correlators he would obtain from this theory would arise from the Fourier transform of world-sheets of volume I to a particle mass $m$. He aims to show that this theory is rigorously defined, and contains perturbative QFT as a sub-branch of it. He hopes that non-perturbative QFT and perhaps other theories will emerge as special limits of simplex correlators for higher dimensional simplices.


\section{Necessary Theorems}
\label{section:def}

This section includes general purpose results required for many portions of this paper. The proofs of these results, unless explicitly proven or cited below, are located in \cite{odwyer2023relativistic}.

\begin{theorem}
\label{thm:multiNomBound}

Let $\{x_{i}\}_{i=1}^{l}\subset\mathbb{N}$ such that $min(x_{i}-\frac{\sum_{i=1}^{l}x_{i}}{l})\ge n$

$$\begin{pmatrix}
\sum_{i=1}^{l}x_{i}\\x_{1},...,x_{l}
\end{pmatrix}= \frac{l^{\sum_{i=1}^{l}x_{i}+\frac{l}{2}}}{\sqrt{2\pi \sum_{i=1}^{l}x_{i}}^{l-1}}e^{-\frac{l}{2\sum_{i=1}^{l}x_{i}}(\sum_{i=1}^{l}(x_{i}-\frac{\sum_{i=1}^{l}x_{i}}{l})^{2})+o(1)}$$

where $o(1)$ denotes some function that goes to zero as $n\rightarrow \infty$.

\end{theorem}

This theorem is given in \cite{344669}

\begin{theorem}
\label{thm:contconv}
The continuous multinomial

$$\begin{Bmatrix}\sum x_{i}\\x_{1},x_{2},...,x_{l}\end{Bmatrix}<\infty$$

is a finite and analytic function from $\mathbb{R}^{l}\rightarrow\mathbb{R}$.

\end{theorem}
\begin{theorem}{The Continuous Multinomial is a Limit of the Discrete One}
\label{thm:disctocont}

Let $\{x_{1}\}_{i=1}^{l}\subset\mathbb{R}_{+}$. Then we have

$$lim_{m\rightarrow\infty}\frac{\mathcal{T}^{m}_{cont}\begin{pmatrix}\sum [m x_{i}]\\ [mx_{1}],..., [mx_{l}]\end{pmatrix}}{\mathcal{T}^{m}_{cont}\begin{pmatrix}\sum_{i}\lfloor nx_{i}\rfloor\\\lfloor \frac{\sum_{i} nx_{i}}{l}\rfloor,...,\lfloor\frac{\sum_{i} nx_{i}}{l}\rfloor\end{pmatrix}}\rightarrow \frac{\begin{Bmatrix}\sum_{i}x_{i}\\x_{1},...,x_{l}\end{Bmatrix}}{\begin{Bmatrix}\sum_{i}x_{i}\\ \frac{\sum_{i} nx_{i}}{l},...,\frac{\sum_{i} nx_{i}}{l}\end{Bmatrix}}$$

This will also prove other forms of convergence (here we normalize by the maxima, we could also normalize by the integral of each expression).

\end{theorem}
\begin{theorem}
\label{thm:conttails}

Let $\{x_{i}\}_{i=1}^{l}\in \mathbb{R}$ where $min(x_{i}-\frac{\sum_{i=1}^{n}x_{i}}{l})>R\in\mathbb{R}_{+}$. Then


$$\begin{Bmatrix}
\sum_{i=1}^{l}x_{i}\\x_{1},...,x_{l}
\end{Bmatrix}= \frac{l^{\sum_{i=1}^{l}x_{i}+\frac{l}{2}}}{\sqrt{2\pi \sum_{i=1}^{l}x_{i}}^{l-1}}e^{-\frac{l}{2\sum_{i=1}^{l}x_{i}}(\sum_{i=1}^{l}(x_{i}-\frac{\sum_{i=1}^{l}x_{i}}{l})^{2})+o(1)}$$

This implies that, should $\sum_{i=1}^{l}x_{i}$ be some constant large number, then its Schwartz class \cite{1972iv} so all derivatives of it have Fourier transforms which are Schwartz class.

\end{theorem}
\begin{theorem}
\label{thm:axes}

Let $d=2$. Then one $\mathcal{A}_{n}$ is the set $\{\frac{1}{gcd(x,I,t)}(I,t)\in\mathbb{N}^{2}|x^{2}+I^{2}=t^{2},(x,I,t)\in\mathbb{Z}^{3},0\le t\le n\}\cup\{(\pm 1,1)\}$ (i.e. a minimal $\mathcal{A}^{gen}$ for $d_{n}$ is this above set).

\end{theorem}
\begin{theorem}

Let $d\ge 2$. Then $\mathcal{A}_{1}=\{(\pm e_{1},1),...,(\pm e_{d},1),(0,1)\}$ where $e_{i}$ denote unit directions in $\mathbb{Z}^{d}$

\label{thm:axes2}

\end{theorem}

These are theorems necessary for the directions in \cite{odwyer2023relativistic}. For our purposes we require the same proof but for the $\mathcal{A}_{n}$ which corresponds to our metric $d_{n}$

\begin{theorem}
\label{thm:highdVersion}
Let $d\ge 2$, Then one $\mathcal{A}_{n}$ for $d_{n}$ is the Equation~\ref{mult:generators}

\begin{multline}
    \label{mult:generators}
    \mathcal{A}_{n}=\{\frac{1}{gcd(\{x_{i}\},t,I)}(\vec{x}_{i},t)\in \mathbb{N}^{d+1}|\sum_{i=1}^{d}x_{i}^{2}+I^{2}=t^{2},(x_{1},...,x_{d},I,t)\in\mathbb{Z}^{d+2},0\le t\le n\}
    \\\cup\{\frac{1}{gcd(\{x_{i}\},t)}(\vec{x}_{i},t)\in \mathbb{N}^{d+1}|\sum_{i=1}^{d}x_{i}^{2}=t^{2},(x_{1},...,x_{d},t)\in\mathbb{Z}^{d+1},0\le t\le n\}
\end{multline}

\end{theorem}
\begin{proof}

Consider the sets $\mathcal{C}_{l^{*}_{2}}=\{d_{l^{*}_{2}}(0,\vec{x})=1\}$ and $\mathcal{C}_{n}=\{d_{n}(0,\vec{x})=1\}$. Let us show that $\mathcal{C}_{l^{*}_{2}}\cap\mathcal{C}_{n}$ is Equation~\ref{eq:pythagTupple}. One may simply evaluate the points Equation~\ref{eq:pythagTupple} and immediately discern is a subset of $\mathcal{C}_{l^{*}_{2}}\cap\mathcal{C}_{n}$ ($\mathcal{C}_{n}$ is carefully designed so this was the case). The first part of the definition of $d_{n}$ is built from hyperplanes interpolating $\vec{v}$ in Equation~\ref{eq:pythagTupple}; these hyperplanes are the sets $\mathcal{H}_{\mathcal{N}_{\vec{v}}}^{-1}(1)$. The second part of the definition includes hyperplanes whose defining direction is a null direction (i.e. along a light path). Using this, we can characterize the set $\mathcal{C}_{n}$ into sets of hyperplanar sections which are compact and interpolate points, and into those planes which are non-compact and orthogonal to a null direction.

The sheets of the hyperboloid $\mathcal{C}_{l_{2}^{*}}$ are individually the images of convex function to the t coordinate (namely $f(x)=\sqrt{1-\sum_{i}x_{i}^{2}}$) and by the definition of convexity any hyperplanar sections formed by adjoining points lying on its surface must all have greater than or equal t values. The equality only occurs if the hessian of our convex function is zero, which is not the case in this setting. Therefore, we immediately have that the compact hyperplanar sections of $\mathcal{C}_{n}$ lie above $\mathcal{C}_{t^{*}}$. The non-compact sections of $\mathcal{C}_{n}$, due to their non-compactness, will lie ougtside the non-compact sections of $\mathcal{C}_{n}$. Their null-direction orthogonal direction implies that, along any direction, they increase larger than any timelike curve in the t-direction as one traces along their surfaces radially. This implies that the non-compact sections of $\mathcal{C}_{n}$ lie properly above $\mathcal{C}_{l_{2}^{*}}$ as well because its orthogonal directions are time-like. These fact together imply that $\mathcal{C}_{l_{2}^{*}}\cap\mathcal{C}_{n}$ is Equation~\ref{eq:pythagTupple}. Consider all directions in $\mathcal{A}^{all}$ with non-zero minkowski metric values; if we divide these vectors by their norm we find a point in $\mathcal{C}_{l_{2}^{*}}\cap\mathcal{C}_{n}$ and therefore in Equation~\ref{eq:pythagTupple}. There is some smallest multiplicative number one may multiply elements of Equation~\ref{eq:pythagTupple} to obtain a vector in $\mathbb{Z}^{d+1}$. The resultant vectors are precisely the non-null portion of Equation~\ref{mult:generators} (as you divide them by their gcd implying they could not have a smaller parallel direction that lies upon the lattice). Therefore, the non-null portion of Equation~\ref{mult:generators} generates all non-null parts of $\mathcal{A}_{all}$ with unique representatives for each direction; we now only need to show the null portion of Equation~\ref{mult:generators} does so for the null directions.

If we now call $\mathcal{B}_{l_{2}^{*}}=\{d_{l_{2}^{*}}(0,\vec{x})=0\}$ and $\mathcal{B}_{n}=\{d_{n}(0,\vec{x})=0\}$, then $\mathcal{A}_{all}\cap\{\vec{v}|d_{l_{2}^{*}}(0,\vec{v})=0\}=\mathcal{B}_{l_{2}^{*}}\cap\mathcal{B}_{n}$. This is a set of rays emanating from the origin; if we show that the null portion of Equation~\ref{mult:generators} has a component or can generate a component aligned with all directions in this set, then we will be done. This is because the division by the gcd in Equation~\ref{mult:generators} implies that there are no smaller vectors parallel to the one in question that also lie upon $\mathbb{Z}^{d}\times \mathbb{Z}$, and they generate all other vectors of the same direction lying upon the lattice.

$\mathcal{B}_{n}$ is going to only feature intersection of hyperplanes given by $\mathcal{H}_{\mathcal{N}_{\vec{a}}}$ for $\vec{a}$ among null directions in $\mathcal{A}_{n}$. These are the same null directions as defined in Equation~\ref{mult:generators}. These hyperplanes interpolate between the rays of each direction in $\mathcal{A}_{n}$. Since our cone $\mathcal{B}_{l_{2}^{*}}$ has some curvature, convexity implies that the hyperplanes interpolating these rays must lie above the conical regions of $\mathcal{B}_{l_{2}^{*}}$ (which also interpolate these rays). Therefore, $\mathcal{B}_{l_{2}^{*}}\cap\mathcal{B}_{n}$ must be the set of rays aligned with directions along $\mathcal{A}_{n}$, and therefore with directions along Equation~\ref{mult:generators}.


\end{proof}

\begin{theorem}
\label{thm:generate}

Let $\gamma=\{x_{i}\}_{i=1}^{n}\in \Gamma_{n}$. Let $\mathcal{A}_{n}$ denote some axes of symmetry of $d_{n}$. Then, we can choose our $\{x_{i}\}^{n}$ such that $x_{i+1}-x_{i}\in\mathcal{A}_{n}$. This alteration will yield the same path up to our equivalence relation on $\Gamma$, and this equivalence class representative is unique.

\end{theorem}
\begin{theorem}
    \label{thm:tripledensity}
    Denote the set of primitive pythagorean quadruples with hypotenuse below $n\in \mathbb{R}$ as $\mathcal{A}_{n}$. Then, $\mathcal{A}_{n}=\frac{n^{2}}{32G}+O(n)$, and they are equidistributed on the unit circle when ordered according to hypotenuse (as the hypotenuse goes to infinity). Here, G denotes Catalan's constant and has a value approximately $\sim .9159655$
\end{theorem}

\begin{proof}

    The proof of the first part of this theorem is found in \cite{Werner2015} and the second in \cite{Duke2003}.
    
\end{proof}
\begin{theorem}
    \label{thm:properties}
    The continuum multinomial has the following property
    $$\begin{Bmatrix}\sum_{i}x_{i}\\x_{1},...,x_{n}\end{Bmatrix}=\int_{-\infty}^{\infty}\begin{Bmatrix}\sum_{i}x_{i}\\x_{1},...,I\end{Bmatrix}\begin{Bmatrix}I\\x_{n-1},x_{n}\end{Bmatrix}dI$$

    This implies that should any of the coefficients of the continuum multinomial coefficient be zero, then the function is also zero.
\end{theorem}

\section*{Acknowledgement}

I want to thank Eviatar Procaccia and Parker Duncan from Technion University for their help in the development of this paper through work in \cite{duncan2020elementary} and \cite{duncan2021discrete}. I thank Hailey Leclerc for her aid in editing this paper, as well as Grigory Rogachev and Lauren Tompkins for their mentor-ship. Finally I thank my parents, Marianne Fairey and Joseph O'Dwyer. I dedicate this paper to the memory of Anthony O'Dwyer.

\newpage

\bibliographystyle{amsplain}
\bibliography{main.bib}

\providecommand{\bysame}{\leavevmode\hbox to3em{\hrulefill}\thinspace}
\providecommand{\MR}{\relax\ifhmode\unskip\space\fi MR }
\providecommand{\MRhref}[2]{%
  \href{http://www.ams.org/mathscinet-getitem?mr=#1}{#2}
}
\providecommand{\href}[2]{#2}
\begin{thebibliography}{10}

\bibitem{bujalance_costa_martinez_2001}
\emph{Topics on riemann surfaces and fuchsian groups}, London Mathematical
  Society Lecture Note Series, Cambridge University Press, 2001.

\bibitem{continuousbin}
Leonardo Cano and Rafael Diaz, \emph{Continuous analogues for the binomial
  coefficients and the catalan numbers}, 2016.

\bibitem{feynman2}
Smith D, \emph{Hyperdiamond feynman checkerboard in 4-dimensional spacetime},
  (1995).

\bibitem{defreitas2022geometrization}
Izabella~Muraro de~Freitas and Álvaro Krüger~Ramos, \emph{Geometrization in
  geometry}, 2022.

\bibitem{Duke2003}
W.~Duke, \emph{Rational points on the sphere}, The Ramanujan Journal \textbf{7}
  (2003), no.~1, 235--239.

\bibitem{duncan2021discrete}
Parker Duncan, Rory O'Dwyer, and Eviatar Procaccia, \emph{Discrete $\ell^{1}$
  double bubble solution is at most ceiling +2 of the continuous solution},
  2021.

\bibitem{duncan2020elementary}
Parker Duncan, Rory O'Dwyer, and Eviatar~B. Procaccia, \emph{An elementary
  proof for the double bubble problem in $\ell^1$ norm}, 2020.

\bibitem{feynman1}
Keith Earle, \emph{Notes on the feynman checkerboard problem},  (2010).

\bibitem{folland}
Gerald~B. Folland, \emph{Real analysis: Modern tichniques and their
  applications}, Wiler-Interscience Publication, 1994.

\bibitem{Hong_Hao_2010}
Zhang Hong-Hao, Feng Kai-Xi, Qiu Si-Wei, Zhao An, and Li~Xue-Song, \emph{On
  analytic formulas of feynman propagators in position space}, Chinese Physics
  C \textbf{34} (2010), no.~10, 1576--1582.

\bibitem{Werner2015}
Werner Hürlimann, \emph{Exact and asymptotic evaluation of the number of
  distinct primitive cuboids}, Journal of Integer Sequences \textbf{18} (2015),
  Article 15.2.5.

\bibitem{344669}
Iosif~Pinelis (https://mathoverflow.net/users/36721/iosif pinelis),
  \emph{Asymptotics of multinomial coefficients}, MathOverflow,
  URL:https://mathoverflow.net/q/344669 (version: 2022-01-02).

\bibitem{10.1119/1.11972}
Jean‐Marc Lévy‐Leblond and Jean‐Pierre Provost, \emph{{Additivity,
  rapidity, relativity}}, American Journal of Physics \textbf{47} (1979),
  no.~12, 1045--1049.

\bibitem{tiechmuller}
Curtis McMullen, \emph{Billiards and teichmüller curves}, Bulletin of the
  American Mathematical Society \textbf{60} (2022).

\bibitem{odwyer2023relativistic}
Rory O'Dwyer, \emph{Relativistic propagators on lattices}, 2023.

\bibitem{pesky}
Michael~Edward Peskin and Daniel~V. Schroeder, \emph{{An Introduction to
  Quantum Field Theory}}, Westview Press, 1995, Reading, USA: Addison-Wesley
  (1995) 842 p.

\bibitem{1972iv}
Micheal Reed, \emph{Methods of modern mathematical physics}, Academic Press,
  1972.

\bibitem{hyperboloidmodel}
William~F. Reynolds, \emph{Hyperbolic geometry on a hyperboloid}, The American
  Mathematical Monthly \textbf{100} (1993), no.~5, 442--455.

\bibitem{Srednicki_2007}
Mark Srednicki, \emph{Quantum field theory}, 1 ed., Cambridge University Press,
  February 2007.

\bibitem{continuouslatticepath}
T.~Wakhare, C.~Vignat, Q.~N. Le, and S.~Robins, \emph{A continuous analogue of
  lattice path enumeration}, 2017.

\bibitem{yifrach2023note}
Yuval Yifrach, \emph{A note about weyl equidistribution theorem}, 2023.

\end{thebibliography}

\end{document}